\documentclass{vldb}
\usepackage{cite,amsmath,amssymb,amsfonts,dsfont,multirow,multicol,epsfig,url,array,makecell,balance,color}
\usepackage[ruled, vlined, linesnumbered]{algorithm2e}
\usepackage{hhline}
\usepackage{array}
\usepackage{enumerate}
\usepackage{enumitem}
\usepackage{times}
\usepackage{subfigure}
\usepackage{booktabs}
\usepackage{epstopdf}
\usepackage{balance}

\newtheorem{theorem}{Theorem}

\newtheorem{lemma}{Lemma}

\newtheorem{definition}{Definition}

\newtheorem{pruning}{Pruning Rule}

\def\header{\vspace{2mm} \noindent}

\def\tblcapup{\vspace{0mm}}
\def\tblcapdown{\vspace{2mm}}
\def\tbldown{\vspace{-0mm}}

\newcommand{\pushright}[1]{\ifmeasuring@#1\else\omit\hfill$\displaystyle#1$\fi\ignorespaces}
\newcommand{\pushleft}[1]{\ifmeasuring@#1\else\omit$\displaystyle#1$\hfill\fi\ignorespaces}

\def\done{\hspace*{\fill} {$\square$}}

\def\e{\varepsilon}
\def\s{\tilde{s}}
\def\O{O}

\def\E{\mathbb{E}}
\def\scw{\sqrt{c}}

\vldbTitle{ProbeSim: Scalable Single-Source and Top-
    k SimRank Computations on Dynamic Graphs}
\vldbAuthors{Yu Liu, Bolong
Zheng,  Xiaodong He,
Zhewei Wei, Xiaokui Xiao, Kai Zheng, Jiaheng Lu}
\vldbVolume{11}
\vldbNumber{1}
\vldbYear{2017}
\vldbDOI{https://doi.org/10.14778/3136610.3136612}
\begin{document}
\begin{sloppy}

\title{ProbeSim: Scalable Single-Source and Top-{\LARGE $\boldsymbol
    k$} SimRank Computations on Dynamic Graphs}

\numberofauthors{1}
\author{\alignauthor
\makebox[6.5in][c] {Yu Liu$^{1}$, Bolong
Zheng$^{2}$,  Xiaodong He$^{1}$,
Zhewei Wei$^{1}$\titlenote{Corresponding author}, Xiaokui Xiao$^{3}$, Kai Zheng$^{4}$, Jiaheng Lu$^{5}$} \\
\affaddr{$^{1}$School of Information, Renmin University of China, China}\\
\affaddr{$^{2}$School of Data and Computer Science, Sun Yat-sen University}\\
\affaddr{$^{3}$School of Computer Science and Engineering, Nanyang
  Technological University, Singapore}\\
\affaddr{$^{4}$School of Computer Science and
Engineering and Big Data Research
Center, University of Electronic
Science and Technology of China}\\
\affaddr{$^5$ Department of Computer Science, University of Helsinki}
\email{$^{1}$\{yu.liu, hexiaodong\_1993, zhewei\}@ruc.edu.cn $\qquad$ $^{2}$zblchris@gmail.com}
\email{$^{3}$xkxiao@ntu.edu.sg $\qquad$
  $^{4}$zhengkai@uestc.edu.cn $\qquad$ $^{5}$jiahenglu@gmail.com}\\
}

\vldbDOI{https://doi.org/10.14778/3136610.3136612}


\maketitle

%

\begin{abstract}
Single-source and top-$k$ SimRank queries are two important types of similarity search in graphs with numerous applications in web mining, social network analysis, spam detection, etc. A plethora of techniques have been proposed for these two types of queries, but very few can efficiently support similarity search over large dynamic graphs, due to either significant preprocessing time or large space overheads. 

This paper presents {\em ProbeSim}, an  {\em index-free}  algorithm for single-source and top-$k$ SimRank queries that provides a non-trivial theoretical guarantee in the absolute error of query results. {\em ProbeSim} estimates SimRank similarities without precomputing any indexing structures, and thus can naturally support {\em real-time} SimRank queries on {\em dynamic} graphs. Besides the theoretical guarantee, {\em ProbeSim} also offers satisfying practical efficiency and effectiveness due to non-trivial optimizations. We conduct extensive experiments on a number of benchmark datasets, which demonstrate that our solutions outperform the existing methods in terms of efficiency and effectiveness. Notably, our experiments include the first empirical study that evaluates the effectiveness of SimRank algorithms on graphs with billion edges, using the idea of {\em pooling}.
\end{abstract}


\section{Introduction} \label{sec:intro}

{\em SimRank} \cite{JW02} is a classic measure of the similarities of graph nodes, and it has been adopted in numerous applications such as web mining \cite{Jin11}, social network analysis \cite{NK07}, and spam detection \cite{SH11}. The formulation of SimRank is based on two intuitive statements: (i) a node is most similar to itself, and (ii) two nodes are similar if their neighbors are similar. Specifically, given two nodes $u$ and $v$ in a graph $G$, the SimRank similarity of $u$ and $v$, denoted as $s(u, v)$, is defined as:
\begin{equation} \label{eqn:intro-simrank}
s(u, v) =
\begin{cases}
1, & \text{if $u = v$}\\
{\displaystyle \frac{c}{|I(u)| \cdot |I(v)|}\hspace{-1mm} \sum_{x \in I(u),y \in I(v)} \hspace{-2mm}{s(x, y)} }, & \text{otherwise.}
\end{cases}
\end{equation}
where $I(u)$ denotes the set of in-neighbors of $u$, and $c \in (0,1)$ is a decay factor typically set to 0.6 or 0.8 \cite{JW02,LVGT10}.

Computing SimRank efficiently is a non-trivial problem that has been studied extensively in the past decade. The early effort \cite{JW02} focuses on computing the SimRank similarities of all pairs of nodes in $G$, but the proposed {\em Power Method} algorithm incurs prohibitive overheads when the number $n$ of nodes in $G$ is large, as there exists $\O(n^2)$ node pairs in $G$. To avoid the inherent $\O(n^2)$ costs in all-pair SimRank computation, the majority of the subsequent work considers two types of SimRank queries instead 
\begin{itemize}[topsep = 6pt, parsep = 6pt, itemsep = 0pt, leftmargin=18pt]

\item Single-source SimRank query: given a query node $u$, return $s(u, v)$ for every node $v$ in $G$;

\item Top-$k$ SimRank query: given a query node $u$ and a parameter $k \ge 1$, return the $k$ nodes $v$ with the largest $s(u, v)$.
\end{itemize}

Existing techniques \cite{KMK14,MKK14,TX16,FRCS05,LeeLY12,SLX15,YuM15b,LiFL15} for these two types of queries, however, suffer from two major deficiencies. First, most methods \cite{MKK14,LeeLY12,LiFL15,SLX15,YuM15b} fail to provide any worst-case guarantee in terms of the accuracy of query results, as they either rely on heuristics or adopt an incorrect formulation of SimRank. Second, the existing solutions \cite{TX16,FRCS05,LiFL15} with theoretical accuracy guarantees, all require constructing index structures on the input graphs with a preprocessing phase, which incurs significant space and pre-computation overheads. The only exception is the Monte Carlo method in \cite{FRCS05} which provides an index-free solution with  theoretical accuracy guarantees. While pioneering, unfortunately, this method entails considerable query overheads, as shown in~\cite{SLX15}.




\vspace{-1mm}
\header
\noindent{\bf Motivations. } In this paper, we aim to develop an {\em index-free} solution for single-source and top-$k$ SimRank queries with provable accuracy guarantees. Our motivation for devising algorithms without preprocessing is two-fold. First, index-based SimRank methods often have difficulties handling {\em dynamic} graphs. For example,  {\em SLING}~\cite{TX16}, which is the state-of-art indexing-based SimRank algorithm for static graphs, requires its index structure to be rebuilt from scratch whenever the input graph is updated, and its index construction requires several hours even on medium-size graphs with 1 million nodes. This renders it infeasible for real-time queries on dynamic graphs. In contrast, index-free techniques can naturally support real-time SimRank queries on graphs with frequent updates. To the best of our knowledge, the {\em TSF} method \cite{SLX15} is the only indexing approach that allows efficient update.  However, {\em TSF} is unable to provide any worst-case guarantees in terms of the accuracy of the SimRank estimations, which leads to unsatisfying empirical effectiveness, as shown in~\cite{zhangexperimental} and in our experiments.


Our second motivation is that index-based SimRank methods often fail to scale to large graphs due to their space overheads. For example, {\em TSF} requires an index space that is two to three orders of magnitude larger than the input graph size, and our empirical study shows that it runs out of 64GB memory for graphs over 1GB in our experiments. This renders it only applicable on small- to medium-size datasets. Further, if one considers to move the large index to the external memory, this idea would incur expensive preprocessing and query costs, as shown in our empirical study. In contrast,  an index-free solution proposed in this paper does not increase the size of an original graph.

\vspace{-1mm}
\header
{\bf Our contributions.} This paper presents an in-depth study on single-source and top-$k$ SimRank queries, and makes the following contributions. First, for single-source and top-$k$ SimRank queries, we propose an algorithm with provable approximation guarantees. In particular, given two constants $\e_a$ and $\delta$, our algorithm ensures that, with at least $1-\delta$ probability, each SimRank similarity returned has at most $\e_a$ absolute error. The algorithm runs in $O(\frac{n}{{\e_a}^2} \log \frac{n}{\delta})$ expected time, and it does not require any index structure to be pre-computed on the input graph. Our algorithm matches the state-of-the-art index-free solution in terms of time complexity, but it offers much higher practical efficiency due to an improved algorithm design and several non-trivial optimizations.




Our second contribution is a large set of experiments that evaluate the proposed solutions with the state of the art on benchmark datasets. Most notably, we present the first empirical study that evaluates the effectiveness of SimRank algorithms on graphs with billion edges, using the idea of {\em pooling} borrowed from the information retrieval community. The results demonstrate that our solutions significantly outperform the existing methods in terms of both efficiency and effectiveness. In addition, our solutions are more scalable than the state-of-the-art index-based techniques, in that they can handle large graphs on which the existing solutions require excessive space and time costs in preprocessing.




\section{Preliminaries} \label{sec:prelim}

\subsection{Problem Definition} \label{sec:prelim-def}
Table~\ref{tbl:def-notation} shows the notations that are frequently used in the remainder of the paper.
Let $G = (V, E)$ be a directed simple graph with $|V| = n$ and $|E| = m$. We aim to answer {\em approximate} single-source and top-$k$ SimRank queries, defined as follows:
\begin{definition}[Approximate Single-Source Queries] \label{def:prelim-single-source}
Given a node $u$ in $G$, an absolute error threshold $\e_a$, and a failure probability $\delta$, an approximate single-source SimRank query returns an estimated value $\s(u, v)$ for each node $v$ in $G$, such that
$$\left\lvert \s(u, v) - s(u, v) \right\rvert \le \e_a$$
holds for any $v$ with at least $1 - \delta$ probability. \done
\end{definition}

\begin{definition}[Approximate Top-$k$ Queries] \label{def:prelim-topk}
Given a node $u$ in $G$, a positive integer $k < n$, an error threshold $\e_a$, and a failure probability $\delta$, an approximate top-$k$ SimRank query returns a sequence of $k$ nodes $v_1, v_2, \ldots, v_k$ and an estimated value $\s(u, v_i)$ for each $v_i$, such that the following equations hold with at least $1 - \delta$ probability for any $i \in [1, k]$:
\begin{align*}
s(u, v_i) \ge s(u, v'_i) - \e_a
\end{align*}
where $v'_i$ is the node in $G$ whose SimRank similarity to $u$ is the $i$-th largest. \done
\end{definition}
Essentially, the approximate top-$k$ query for node $u$ returns $k$ nodes $v_1, \ldots, v_k$ such that their {\em actual} SimRank similarities with respect to $u$ are $\e$-close to those of the actual top-$k$ nodes. It is easy to see that an approximate single-source algorithm can be extended to answer the approximate top-$k$ queries, by sorting the SimRank estimations $\{\s(u, v) \mid v \in V\}$ and output the top-$k$ results. Therefore, our main focus is on designing efficient and scalable algorithms that answer approximate single-source queries with $\e_a$ guarantee.


\subsection{SimRank with Random Walks} \label{sec:prelim-randomwalk}

In the seminal paper \cite{JW02} that proposes SimRank, Jeh and Widom show that there is an interesting connection between SimRank similarities and random walks. In particular, let $u$ and $v$ be two nodes in $G$, and $W(u)$ (resp.\ $W(v)$) be a random walk from $u$ that follows a randomly selected incoming edge at each step. Let $t$ be the smallest positive integer $i$ such that the $i$-th nodes of $W_u$ and $W_v$ are identical. Then, we have
\begin{equation} \label{eqn:prelim-rw1}
s(u, v) = \E[c^{t-1}],
\end{equation}
where $c$ is the decay factor in the definition of SimRank (see Equation~\ref{eqn:intro-simrank}).

Subsequently, it is shown in \cite{TX16} that Equation~\ref{eqn:prelim-rw1} can be simplified based on the concept of {\em $\scw$-walks}, defined as follows.
\begin{definition}[$\scw$-walks] \label{def:prelim-scw}
Given a node $u$ in $G$, an $\scw$-walk from $u$ is a random walk that follows the incoming edges of each node and stops at each step with $1-\scw$ probability. \done
\end{definition}
A $\scw$-walk from $u$ can be generated as follows. Starting from $v=u$, when
 visiting node $v$, we generate a random number $r$ in $[0,1]$ and check whether $r \le 1-  \sqrt{c}$. If so, we terminate the walk at $v$; otherwise, we select one of the in-neighbors of $v$ uniformly at random and proceed to that node.

Let $W'(u)$ and $W'(v)$ be two $\scw$-walks from two nodes $u$ and $v$, respectively. We say that two $\scw$-walks {\em meet}, if there exists a positive integer $i$ such that the $i$-th nodes of $W'(u)$ and $W'(v)$ are the same. Then, according to \cite{TX16},
\begin{equation} \label{eqn:prelim-rw2}
s(u, v) = \Pr\left[\textrm{$W'(u)$ and $W'(v)$ meet}\right].
\end{equation}
Based on Equation~\ref{eqn:prelim-rw2}, one may estimate $s(u, v)$ using a Monte Carlo approach \cite{FRCS05,TX16} as follows. First, we generate $r$ pairs of $\scw$-walks, such that the first and second walks in each pair are from $u$ and $v$, respectively. Let $r_{meet}$ be the number of $\scw$-walk pairs that meet. Then, we use $r_{meet}/r$ as an estimation of $s(u, v)$. By the Chernoff bound, it can be shown that when
$r \ge \frac{1}{2{\e_a}^2}\mathrm{log}\frac{1}{\delta},$
with at least $1 - \delta$ probability we have
$\left| \frac{r_{meet}}{r} - s(u, v) \right| \le \e_a.$
In addition, the expected time required to generate $r$ $\scw$-walks is $O(r)$, since each $\scw$-walk has $\frac{1}{1-\scw}$ nodes in expectation.

The above Monte Carlo approach can also be straightforwardly adopted to answer any approximate single-source SimRank query from a node $u$. In particular, we can generate $r$ $\scw$-walks from each node, and then use them to estimate $s(u, v)$ for every node $v$ in $G$. This approach is simple and does not require any pre-computation, but it incurs considerable query overheads, since it requires generating a large number of $\scw$-walks from each node.

\subsection{Competitors}
\label{sec:compare}
\vspace{2mm}
\noindent{\bf The TopSim based algorithms.} To address the drawbacks of the Monte Carlo approach, Lee et al.\ \cite{LeeLY12} propose {\em TopSim-SM}, an index-free algorithm that answers top-$k$ SimRank queries by enumerating all short random walks from the query node.  More precisely, given a query node $u$ and a number $T$, {\em TopSim-SM} enumerates all the vertices that reach $u$ by at most $T$ hops, and treat them as potential meeting points. Then, {\em TopSim-SM} enumerates, for each meeting point $w$, the vertices that are reachable from $w$ within $T$ hops. Lee et al.\  \cite{LeeLY12} also propose two variants of {\em TopSim-SM}, named {\em Trun-TopSim-SM} and {\em Prio-TopSim-SM}, which trade accuracy for efficiency. In particular,  {\em Trun-TopSim-SM} omits the meeting points with large degrees, while {\em Prio-TopSim-SM} prioritizes the meeting points in a more sophisticated manner and explore only the high-priority ones.


For each node $v$ returned, {\em TopSim-SM} provides an estimated SimRank $s_T(u, v)$ that equals the SimRank value approximated using the {\em Power Method} \cite{JW02} with $T$ iterations. When $T$ is sufficiently large, $s_T(u, v)$ can be an accurate approximation of $s(u, v)$. However, Lee et al.\ \cite{LeeLY12} show that the query complexity of {\em TopSim-SM} is $O(d ^{2T})$ time, where $d$ is the average in-degree of the graph. As a consequence, Lee et al.\ \cite{LeeLY12} suggests setting $T = 3$ to achieve reasonable efficiency, in which case the absolute error in each SimRank score can be as large as $c^3$, where $c$ is the decay factor in the definition of SimRank \ref{eqn:intro-simrank}. Meanwhile, {\em Trun-TopSim-SM} and {\em Prio-TopSim-SM} does not provide any approximation guarantees even if $T$ is set to a large value, due to the heuristics that they apply to reduce the number of meeting points explored.




\vspace{2mm}
\noindent{\bf The {\em TSF} algorithm.} Very recently, Shao et al.\ \cite{SLX15} propose a two-stage random-walk sampling framework ({\em TSF}) for top-$k$ SimRank queries on dynamic graphs. Given a
parameter $R_g$, {\em TSF} starts by building $R_g$ {\em one-way graphs} as an index structure.  Each one-way graph is constructed by uniformly sampling \emph{one} in-neighbor from each vertex's in-coming edges. The one-way graphs are then used to simulate random walks during query processing.

To achieve high efficiency, {\em TSF} approximates the SimRank score of two nodes $u$ and $v$ as
$$\sum_{i} \Pr[\textrm{two $\sqrt{c}$-walks from $u$ and $v$ meet at the $i$-th step}],$$
which is an over estimation of the actual SimRank. (See Section 3.3 in \cite{SLX15}.) Furthermore, {\em TSF} assumes that every random walk in a one-way graph would not contain any cycle, which does not always hold in practice, especially for undirected graphs. (See Section 3.2 in \cite{SLX15}.) As a consequence, the SimRank value returned by {\em TSF} does not provide any theoretical error assurance.


\begin{table} [t]
\centering
\renewcommand{\arraystretch}{1.3}
\begin{small}
\tblcapup
\caption{Table of notations.}\label{tbl:def-notation}
\tblcapdown
 \begin{tabular} {|l|p{2.5in}|} \hline
   {\bf Notation}       &   {\bf Description}                                       \\ \hline
   $G$         &   the input graph                                                  \\ \hline
   $n, m$      &   the numbers of nodes and edges in $G$                            \\ \hline
   $I(v)$       &   the set of in-neighbors of a node $v$ in $G$                           \\ \hline
   $s(u, v)$    &   the SimRank similarity of two nodes $u$ and $v$ in $G$     \\ \hline
   $\s(u, v)$    & an estimation of $s(u, v)$                                    \\ \hline
   $W(u)$    & a $\scw$-walk from a node $u$                                    \\ \hline
   $c$          &   the decay factor in the definition of SimRank                   \\ \hline
   $\e_a$         &   the maximum absolute error allowed in SimRank computation          \\ \hline
   $\delta$     &   the failure probability of a Monte Carlo algorithm              \\ \hline
 \end{tabular}
\end{small}
\end{table}

\section{ProbeSim Algorithm} \label{sec:single}
\renewcommand{\E}{\mathrm{E}}
In this section, we present {\em ProbeSim}, an index-free
algorithm for approximate single-source and top-$k$ SimRank queries on
large graphs. Recall that an approximate single-source algorithm can
be extended to answer the approximate top-$k$ queries, by sorting the
SimRank estimations $\{\s(u, v) \mid v \in V\}$ and output the top-$k$
results. Therefore, the {\em ProbeSim} algorithm described in this section focuses on approximate single-source queries with $\e_a$ guarantee.
Before diving into the details of the algorithm, we
first give some high-level ideas of the algorithm.
\subsection{Rationale} \label{sec:single-over}
Let $W(u)$ and $W(v)$ be two $\scw$-walks from two nodes $u$ and $v$, respectively. Let $u_i$ be the $i$-th node in $W(u)$. (Note that $u_1 = u$.) By Equation~\ref{eqn:prelim-rw2},
\begin{align} \label{eqn:single-simrank-meet}
s(u, v) & = \Pr\left[\textrm{$W(u)$ and $W(v)$ meet}\right] \nonumber \\
& = \sum_{i} \Pr\left[\textrm{$W(u)$ and $W(v)$ first meet at $u_i$}\right].
\end{align}
In other words, for a given $W(u)=(u_1, u_2, \ldots)$, if we can estimate the probability that an $\scw$-walk from $v$ first meets $W(u)$ at $u_i$, then we can take the sum of the estimated probabilities over all $u_i$ as an estimation of $s(u, v)$. Towards this end, a naive approach is to generate a large number of $\scw$-walks from $v$, and then check the fraction of walks that first meet $W(u)$ at $u_i$. However, if $s(u, v)$ is small, then most of the $\scw$-walks is ``wasted'' since they would not meet $W(u)$. To address this issue, our idea is as follows: instead of sampling $\scw$-walks from each $v$ to see if they can ``hit'' any $u_i$, we start a graph traversal from each $u_i$ to identify any node $v$ that has a non-negligible probability to ``walk'' to $u_i$. Intuitively, this significantly reduces the computation cost since it may enable us to omit the nodes whose SimRank similarities to $u$ are smaller than a given threshold $\e_a$.

In what follows, we first explain the details of the traversal-based algorithm mentioned above, and then analyze its approximation guarantee and time complexity. For convenience, we formalize the concept of {\em first-meeting probability} as follows.
\begin{definition}[first-meeting probability]
\label{def:first_meet}
Given a reverse path $\mathcal{P}= (u_1,\ldots, u_i)$ and a node $v
\in V$, $v \neq u_1$, the first-meeting probability of $v$ with respect
to $\mathcal{P}$ is defined to be
$$P(v, \mathcal{P}) = \Pr_{W(v)}[v_i = u_i, v_{i-1} \neq u_{i-1},
\ldots, v_1 \neq u_1],$$
where  $W(v) = (v_1, \ldots, v_i, \ldots)$ is a random $\sqrt{c}$-walk
that starts at $v_1 =v$.
\end{definition}

Here, the subscript in $\Pr_{W(v)}$ indicates that the randomness arises from the choices of $\sqrt{c}$-walk $W(v)$. In the remainder of the paper, we will omit this subscript when the context is clear. 

\subsection{Basic algorithm} \label{sec:single-basic}
We describe our basic algorithm for the {\em ProbeSim} algorithm.
Given a node $u \in V$, a {\em sampling
error} parameter $\varepsilon$ and a failure probability $\delta$, the algorithm returns
$\mathcal{R}$, an hash\_set of $n-1$ nodes in $V \setminus \{u\}$ and their SimRank
estimations.
 For {\em EVERY} node $v \in
V$, $v \neq u$, algorithm~\ref{alg:single_source} returns an
estimated SimRank $\s(u, v)$ to the actual SimRank $s(u, v)$  with guarantee
${\Pr[|\s(u,v) - s(u, v)|\le \varepsilon] \ge 1 - \delta}$. Note that
the basic algorithm uses unbiased sampling to produce the estimators,  thus we can set
$\varepsilon_a = \varepsilon$.

\begin{algorithm}[t]
\begin{small}
\caption{Basic {\em ProbeSim} algorithm\label{alg:single_source}}
\KwIn{Directed graph $G=(V,E)$; $u \in V$; Error $\varepsilon$ and failure probability $\delta$ \\}
\KwOut{$\mathcal{R}=\{(v, \s(u,v)) \mid v \in V\}$, a hash\_set of
  size $n$ that maintains the SimRank estimations for
  each node $v \in V$}
$n_r \leftarrow
\frac{3c}{{\varepsilon}^2}\mathrm{log}\frac{n}{\delta}$; \\
\For{ $k=1$ to $n_r$}
{
Generate $\sqrt{c}$-walk $W_k(u)= (u=u_1, \ldots, u_\ell)$; \\
Initialize hash\_set $\mathcal{H}$;\\
\For{ $i= 2, \ldots, \ell$}
{
Set hash\_set $\mathcal{S} \leftarrow \mathsf{PROBE}((u_1, \ldots, u_i))$;\\
\For{each $(v, Score(v)) \in \mathcal{S}$}
{
\If{$(v, \s_k(u,v)) \in \mathcal{H}$}
{$\s_k(u,v) \leftarrow \s_k(u,v) +
 Score(v)$;}
\Else
{Insert $(v,  Score(v))$ to $\mathcal{H}$;}

}
}
\For{ each $(v, \s(u, v)) \in \mathcal{R}$}
{
\If{$\exists (v, \s_k(u, v)) \in \mathcal{H}$}
{
 $\s(u,v) \leftarrow \s(u,v)\cdot {k-1 \over k} + \s_k(u,v) \cdot {1
   \over k}$;
}
\Else
{
 $\s(u,v) \leftarrow \s(u,v)\cdot {k-1 \over k}$;
}
}
}

\Return $\mathcal{R}$;
\end{small}
\end{algorithm}


 The pseudo-code for the basic {\em ProbeSim} algorithm is
 illustrated in
Algorithm~\ref{alg:single_source}. The algorithm runs $n_r =
\frac{3c}{{\varepsilon}^2}\mathrm{log}\frac{n}{\delta}$ independent
trials (Line 1). For the $k$-th trial, the algorithm generates a
$\sqrt{c}$-walk $W_k(u) = (u=u_1, \ldots, u_\ell)$ (Lines
2-3), and invokes
the \textsf{PROBE} algorithm on partial $\sqrt{c}$-walk $W_k(u, i) = (u_1,
\ldots, u_i)$ for $i = 2, \ldots, \ell$ (Lines 5-6). The \textsf{PROBE} algorithm
computes a $Score(v)$ for each node $v \in V$. As we shall
see later, $Score(v)$  is equal to $P(v, W_k(u,i))$, the
first-meeting probability of $v$ with
respect to partial walk $W_k(u,i)$. Let $Score_i(v)$ denote the score computed by the
\textsf{PROBE} algorithm on partial $\sqrt{c}$-walk $W_k(u, i) $, for $i = 2, \ldots, \ell$.
The algorithm sums up all scores to form the estimator $\s_k(u, v)
=\sum_{i=2}^\ell Score_i(v)$ (Lines 7-11).

\begin{algorithm}[h]
\begin{small}
\caption{Deterministic $\mathsf{PROBE}$ algorithm\label{alg:Probe}}
\KwIn{A partial $\sqrt{c}$-walk $(u=u_1, \ldots,
  u_i)$}
\KwOut{$\mathcal{S}=\{(v, Score(v)) \mid  v \in V\}$, a hash\_set of
  nodes and their scores w.r.t. partial walk $W(u,i)$}
Initialize hash\_set $\mathcal{H}_j$ for $j=0,\ldots, i-1$;\\
Insert $(u_i,1)$ to $\mathcal{H}_0$; \\
\For{ $j=0$ to $i-2$}
{
\For {each $(x,Score(x))\in \mathcal{H}_j$}
{
\For {each $v \in \mathcal{O}(x)$ and $v \neq u_{i-j-1}$}
{
\If{$(v, Score(v)) \in \mathcal{H}_{j+1}$}
{$Score(v)  \leftarrow Score(v) +
  \frac{\sqrt{c}}{|I(v)|}  \cdot Score(x)$;}
\Else
{Insert $(v,\frac{\sqrt{c}}{|I(v)|}  \cdot Score(x))$ to $\mathcal{H}_{j+1}$;}
}
}
}
\Return $\mathcal{S} = \mathcal{H}_{i-1}$\;
\end{small}
\end{algorithm}
Finally, for each node $v$, we take the average over the $n_r$ independent
estimators to form the final estimator $\s(u,v) = {1\over n_r}
\sum_{k=1}^{n_r}\s_k(u,v)$. Note that if we take the average after all
$n_r$ trials finish, it would require $\Omega(n_r \cdot n)$ space
to store all the $\s_k(u,v)$ values. Thus, we dynamically update the average estimator
$\s(u,v) \in \mathcal{R}$ for each trial (Lines 12-16). After all $n_r$ trials
finishes, we return $\mathcal{R}$ as the SimRank estimators for each node $v
\in V$ (Line 17).


\vspace{2mm}
\noindent{\bf Deterministic \textsf{PROBE} Algorithm.}
We now give a simple deterministic \textsf{PROBE} algorithm for
computing the scores in Algorithm~\ref{alg:single_source}.
Given a partial $\sqrt{c}$-walk $W(u, i) =(u_1, \ldots, u_i)$ that
starts at $u=u_1$, the \textsf{PROBE} algorithm
outputs $\mathcal{S}=\{(v, Score(v)) \mid  v \neq u \in V\}$, a hash\_set of
nodes and their first-meeting probability with respect to reverse
path $W(u, i)$.

The pseudo-code for the algorithm is shown in
Algorithm~\ref{alg:Probe}. The algorithm initializes $i-1$ hash
  tables $\mathcal{H}_0, \ldots, \mathcal{H}_{i-1}$ (Line 1) and adds
  $(u_i, 1)$ to $\mathcal{H}_0$ (Line 2). In the
  $j$-th iteration, for each node $x$ in $\mathcal{H}_{j}$, the algorithm
  finds each out-neighbour $v \in \mathcal{O}(x)$, and checks
  if $(v, Score(v)) \in \mathcal{H}_{j+1}$ (Lines 3-5). If so, the algorithm  adds ${\sqrt{c}
    \over |I(v)|}\cdot Score(x)$ to $Score(v)$ (Lines 6-7). Otherwise,
  it adds $(v, \frac{\sqrt{c}}{|I(v)|}  \cdot
  Score(x))$ to $\mathcal{H}_{j+1}$ (Lines 8-9). We note that in this
  iteration, no score is added to $u_{i-j-1}$, which ensures
  that the walk $W(v)$ avoids $u_{i-j-1}$ at $v_{i-j-1}$ (Line 5).

The intuition of the \textsf{PROBE} algorithm is as follows. For the
ease of presentation, we let $Score(v, j)$ denote the score computed on the $j$-th
  iteration for node $v \in V$. One can show that $Score(v, j)$ is in fact equal to $P(v, (u_{i-j-1}, \ldots, u_i
 ) )$, the first-meeting probability of each
 node $v$ with respect to reverse path $(u_{i-j-1}, \ldots, u_i
 )$ . Consequently, after the $(i-1)$-th iteration, the algorithm
 computes $Score(v, i-1) =P(v, (u_{1}, \ldots, u_i
 ) )$, the first-meeting probability of each
 node $v$ with respect to reverse path $W(u, i)$. We will make this
 argument rigorous in the analysis.


\noindent{\bf Running Example for Algorithm~\ref{alg:single_source} and~\ref{alg:Probe}.}
Throughout the paper, we will use a toy graph in Figure
\ref{fig:examplegraph} to illustrate our algorithms and pruning
rules. For ease of presentation, we set the decay factor $c' = 0.25$
so that $\sqrt{c'} =0.5$. The SimRank values of each node to $a$ is
listed in  Table \ref{tbl:simrankscoreab}, which are computed by the Power Method
within $10^{-5}$ error.
\begin{table}[t]
\centering
\linespread{1.1}
\begin{small}
\vspace{-2mm} \caption{SimRank similarities with respect to node $a$.} 
\label{tbl:simrankscoreab}
\begin{tabular}{|@{ }c@{ }|@{ }c@{ }|@{ }c@{ }|@{ }c@{ }|@{ }c@{ }|@{ }c@{ }|@{ }c@{ }|@{ }c@{ }|@{ }c@{ }|}
\hline
$ $ & a & b & c & d & e & f & g & h  \\
\hline
$s(a,*)$ & 1.0 & 0.0096 & 0.049 & 0.131 & 0.070 & 0.041 & 0.051 & 0.051  \\
\hline
\end{tabular}
\end{small}

\end{table}


Suppose at the $k$-th trial,  a random $\sqrt{c}$-walk
$W(a)=(a_1, a_2, a_3, a_4) =(a,b, a,b)$ is
generated. Figure~\ref{fig:probe} illustrates the traverse process of the
deterministic \textsf{PROBE} algorithm.
For simplicity,
we only demonstrate the traverse process for partial walk $W(a, 4) = (a_1, a_2, a_3,
a_4)= (a,b,a,b)$, which is represented by the  right-most tree in Figure~\ref{fig:probe}. The algorithm first inserts $(b, 1)$ to $\mathcal{H}_0$. Following the out-edges
of $a_4 = b$, the algorithm finds $a$ and omits it as
$a_3=a$. Next, the algorithm finds $c$,
computes $Score(c, 1) = Score(b, 0)\cdot {\sqrt{c'} \over |I(c)|} = 1 \cdot
{0.5 \over 3} =0.167$, and insert $(c, 0.167)$ to
$\mathcal{H}_1$. Similarly, the algorithm finds $d$ and $e$, and
inserts $Score(d, 1)=\frac{0.5}{1}=0.5$ and
$Score(e, 1)=\frac{0.5}{2}=0.25$ to $\mathcal{H}_1$. For the next
iteration, we find $a$, $f$, $g$ and $h$ from the out-neighbours of $c$,
$d$ and $e$. Note that $b$ is omitted due to the fact that $a_2 =
b$. The score of $f$ at this iteration is computed by
$Score(f, 2) = (Score(c, 1)
+Score(d, 1)+ Score(e, 1) )\cdot {\sqrt{c'} \over |I(f)|}
= (0.167+0.5+0.25) \cdot {0.5 \over4}= 0.115.$

Similarly, the algorithm computes $Score(a, 2) = 0.042$,
$Score(g, 2) = 0.153$ and $Score(h, 2) = 0.153$, and insert $(a, 0.042)$, $(f, 0.115)$, $(g,
0.153)$ and $(h,0.153)$ to $\mathcal{H}_2$. Finally, for the last
iteration, the algorithm computes $Score(b, 3)=0.011,
Score(c,3)=0.033, Score(e, 3)=0.038$ and $Score(f, 3)=0.019$, and returns
$\mathcal{H}_{3} = \{(b, 0.011), (c, 0.033), (e, 0.038), (f, 0.019)\}$
as the results.


\begin{figure}[t]
\hspace{-5mm}
\begin{minipage}[t]{0.40\linewidth}
\centering
\includegraphics[width=25mm]{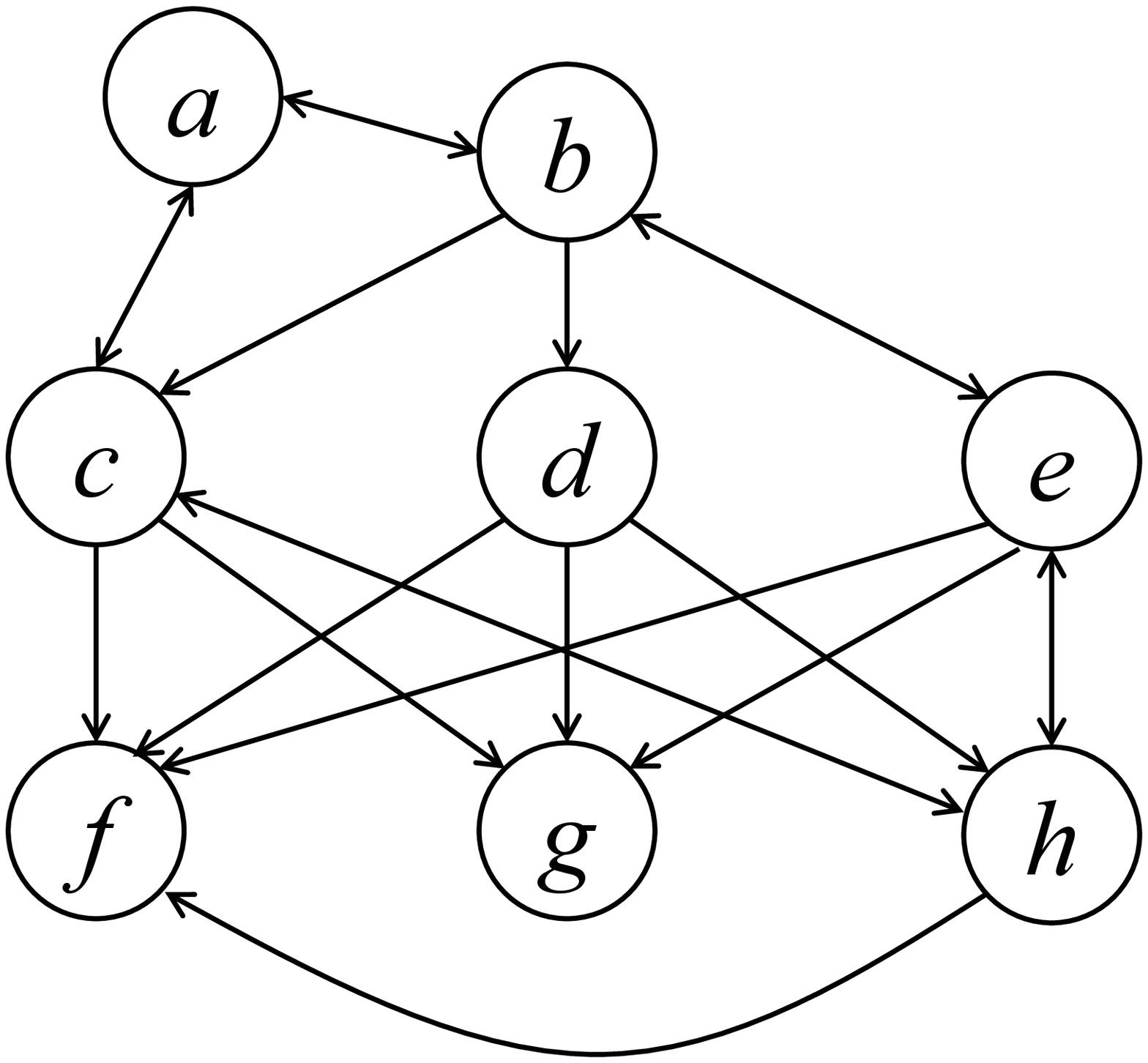}
\vspace{-2mm}
\caption{Toy graph.}
\label{fig:examplegraph}
\end{minipage}%
\hspace{-1mm}
\begin{minipage}[t]{0.65\linewidth}
\centering
\includegraphics[width=60mm]{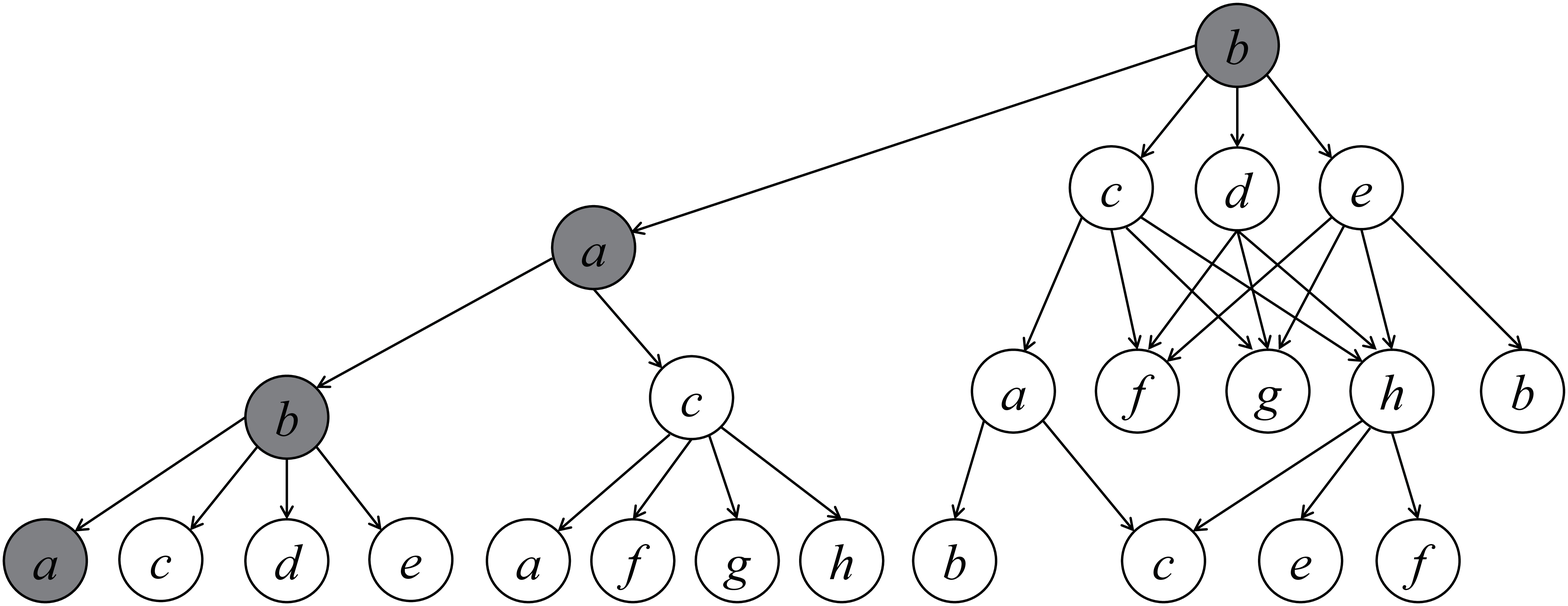}
\vspace{-5.5mm}
\caption{ The probing tree.}
\label{fig:probe}
\end{minipage}
\end{figure}

To get an estimation from $\sqrt{c}$-walk
$W(a)=(a_1, a_2, a_3, a_4) =(a,b,a,b)$,
Algorithm~\ref{alg:single_source} will invoke \textsf{PROBE} for
$W(u,2) = (a,b)$, $W(u,3) = (a,b,a)$ and $W(u,4) = (a,b,a,b)$. Each probe gives score set: $\mathcal{S}_2
= \{(c, 0.167), (d, 0.5), (e, 0.25))\}$, $\mathcal{S}_3=
\{(f, 0.021), (g, 0.028), (h, 0.028)\}$ and
$\mathcal{S}_4
=\{(b, 0.011), (c, 0.033), (e, 0.038), (f, 0.019)\}$. As an example,
the estimator $\s(a,c)$ is computed by summing up all scores for
$c$, which equals to $ 0.167 + 0.033 =0.2$.
By summing all scores up from different nodes
in $W(a)$, the returned estimation of SimRank scores are
$\tilde{s}(a,b)=0.011, \tilde{s}(a,c)=0.2, \tilde{s}(a,d)=0.5,
\tilde{s}(a,e)=0.2877, \tilde{s}(a,f)=0.04, \tilde{s}(a,g)=0.028$ and
$\tilde{s}(a, h)=0.028$. \\



\subsection{Analysis}
\label{sec:analysis_basic}
\vspace{2mm}
\noindent{\bf Time Complexity.}
We notice that in each iteration of the \textsf{PROBE} algorithm,
each edge in the graph is traversed at most once. Thus the time
complexity of the \textsf{PROBE} algorithm is $O(m \cdot  i)$, where $i$ is
the length of the partial $\sqrt{c}$-walk $W(u, i)$. Consequently, the
expected time
complexity of probing a single $\sqrt{c}$-walk  in
Algorithm~\ref{alg:single_source} is bounded by
$O(\sum_{i=1}^{\ell} i\cdot m ) = O(\ell^2 m),$
where $\ell$ is the length of the $\sqrt{c}$-walk $W(u)$.
We notice that each step in the $\sqrt{c}$-walk
terminates with probability at least $1-\sqrt{c}$, so $\ell$ is
bounded by a geometric distributed random variable $X$ with successful probability $p= 1-
\sqrt{c}$ (here ``success'' means the termination of the
$\sqrt{c}$-walk). It follows that
\begin{align*}
\E[\ell^2] &\le E[X^2] =\mathrm{Var}(X) +\E[X]^2 = {1-p \over p^2} +{ 1 \over p^2} \\
&= {2-p \over p^2} = {1+\sqrt{c} \over
  (1-\sqrt{c})^2}  =O(1).
\end{align*}
 Therefore, the
expected running time of Algorithm~\ref{alg:single_source} on a
single $\sqrt{c}$-walk is $O(m)$.
Summing up for $n_r$ walks follows that the expected running time of
Algorithm~\ref{alg:single_source}  is bounded by $O(m n_r) = O({m \over \varepsilon^2} \log {n \over \delta})$.

\vspace{2mm}
\noindent{\bf Correctness.}
We now show that Algorithm~\ref{alg:single_source} indeed gives an
good estimation to the SimRank values $s(u,v)$ for each $v \in V$,
$v\neq u$.
The following Lemma states that each trial
in Algorithm~
\ref{alg:single_source} gives an unbiased estimator for the SimRank value
$s(u ,v)$.


\begin{lemma}
For any $v \in V$ and $v \neq u$, Algorithm
\ref{alg:single_source} gives an estimator $\s(u,v)$ such that
$\E[\s(u, v)]=s(u,v)$.
\label{lem:unbiasness}
\end{lemma}

We need the following Lemma, which states that if we start a
$\sqrt{c}$-walk $W(v)= (v_1, v_2, \ldots)$, then the score computed by
the \textsf{PROBE} algorithm on partial walk $W(u, i) = (u_1,\ldots,
u_i)$ is exactly the probability that $W(u)$ and $W(v)$ first meet at $v_i =
u_i$.

\begin{lemma}
For any node $v \in V$, $v \neq u$, after the $(i-1)$-iteration,
$Score(v, i)$
is equal to $\Pr[v, W(u, i)]$, the first-meeting probability of $v$ with
respect to partial $\sqrt{c}$-walk $W(u, i)$.
\label{lem:Score}
\end{lemma}

\begin{proof}
We prove the following claim: Let $Score(v, j)$ denote the score of
$v$ after the $j$-th iteration. Fix a node $v\in V$, $v \neq u$. After the $j$-th iteration in
Algorithm~\ref{alg:Probe}, we have $Score(v, j) = P(v, (u_{i-j}, \ldots,
u_i))$, the first-meeting probability of $v$ with respect to reverse
path $(u_{i-j}, \ldots, u_i)$. Recall that
$$P(v, (u_{i-j}, \ldots, u_i)) \hspace{-1mm}= \hspace{-1mm}\Pr_{W(v)}[v_{j+1} = u_i, \hspace{-1mm}v_{j} \neq
u_{i-1},
\ldots, v_1 \neq u_{i-j}],$$
where  $W(v) = (v_1, \ldots, v_{j+1}, \ldots)$ is a random $\sqrt{c}$-walk
that starts at $v_1 =v$.

Note that if above claim is true, then after the $(i-1)$-th iteration, we have $Score(v) =
Score(v, i-1) = P(v, (u_1, \ldots, u_i)) =P(v, W(u, i))$, and the Lemma
will follow.
We prove the claim by induction. After the $0$-th iteration, we have
$Score(u_i, 0) =1$ and $Score(v, 0) =0 $ for $v \neq u_i$, so the claim
holds. Assume the claim holds for the $(j-1)$-th iteration. After the
$j$-th iteration, for each $v \in V$, $v \neq u_{i-j-1}$, the
algorithm set $Score(v, j+1)$ by equation
\vspace{-1mm}
\begin{align}
Score(v, j+1) = \sum_{x \in I(v)}{\sqrt{c} \over |I|}
\cdot Score(x, j). \label{eqn:probe_Score}
\end{align}

By the induction hypothesis we have $Score(x, j) = P(x, (u_{i-j}, \ldots,
u_i))$, and thus
\vspace{-1mm}
\begin{align}
&  \quad \quad Score(v, j+1) = \sum_{\substack{x \in I(v)\\x\neq u_{i-j+1} }}{\sqrt{c} \over |I|} \cdot  P(x, (u_{i-j+1}, \ldots,
u_i)) \nonumber\\
&= \sum_{\substack{x \in I(v)\\x\neq u_{i-j+1} }} \Pr[v_2 =x ] \cdot P(x, (u_{i-j+1}, \ldots,
u_i)). \label{eqn:probe1}
\vspace{-3mm}
\end{align}
Here $\Pr[v_2=x]$ denotes the probability that $W(v)$ selects $x$ at the
  first step.
On the other hand, $P(v, (u_{i-j}, \ldots, u_i))$ can be expressed the
summation of probabilities that $W(v)$ first select a node $x \in
I(v)$ that is not $u_{i-j+1}$, and then select a reverse $\sqrt{c}$-walk from $x$ to
$u_i$ of length $j-2$ and avoid $u_{i-j+k}$ at $k$-th step. It follows
that
\vspace{-1mm}
\begin{align}
  & \quad \quad P(v, (u_{i-j}, \ldots, u_i))  \nonumber\\
&= \sum_{\substack{x \in I(v)\\x\neq u_{i-j+1} }} \vspace{-4mm}
  \Pr[v_2 =x] \cdot \Pr[v_{j+1} = u_i,
\ldots, v_2 \neq u_{i-j+1}] \nonumber \\
&= \sum_{\substack{x \in I(v)\\x\neq u_{i-j+1} }} \vspace{-4mm}\Pr[v_2 =x ] \cdot P(x, (u_{i-j+1}, \ldots,
u_i)).  \label{eqn:probe2}
\end{align}

Combining equations~\eqref{eqn:probe1} and \eqref{eqn:probe2} proves
the claim, and the Lemma
follows.
\end{proof}

With the help of Lemma~\ref{lem:Score}, we can prove Lemma~\ref{lem:unbiasness}:
\begin{proof}[of Lemma~\ref{lem:unbiasness}]
Let $\mathbf{W}(u)$ denote the set of all possible $\sqrt{c}$-walks
that starts at $u$. Fix a $\sqrt{c}$-walk $W(u) = (u_1,
\ldots, u_\ell)$.  By Lemma~\ref{lem:Score},
 the estimated SimRank can be expressed as $\s(u ,v) = \sum_{i=2}^\ell
P(v, W(u, i))$. Thus we can compute the expectation of this estimation
 by
\vspace{-3mm}
\begin{align}
\E[\s(u, v)] &= \sum_{W(u) \in \mathbf{W}(u)}\Pr[W(u)] \cdot
\sum_{i=2}^{\ell} P(v, W(u, i)) \nonumber \\
&= \sum_{W(u) \in \mathbf{W}(u)}\sum_{i=2}^\ell \Pr[W(u)] \cdot P(v, W(u,
  i)),\label{eqn:unbaisness1}
\end{align}
where $\Pr[W(u)]$ is the probability of walk $W(u)$.
Recall  that $P(v, W(u, i))$ is the probability that a random
$\sqrt{c}$-walk $W(v) = (v_1, \ldots, v_i,  \ldots)$ first meet $W(u)$
at $u_i = v_i$. For the ease of presentation, let $I(W(u), W(v), i)$
denote that indicator variable that two
$\sqrt{c}$-walk $W(v) = (v_1, \ldots, v_i,  \ldots)$ and $W(u) = (u_1, \ldots, u_i,  \ldots)$ first meet
at $u_i = v_i$. In other word, $I(W(u), W(v), i) =1$ if $W(u)$ and
$W(v)$ first meet at $u_i = v_i$, and  $I(W(u), W(v), i) =0$ if otherwise. We have
\begin{equation}
\label{eqn:unbaisness2}
P(v, W(u,
  i)) =\sum_{W(v)
  \in \mathbf{W}(v)} \Pr[W(v)]\cdot I (W(u), W(v), i).
\end{equation}
Combining equation~\eqref{eqn:unbaisness1}
and~\eqref{eqn:unbaisness2}, it follows that
\begin{align*}
\E[\s(u, v)] &= \hspace{-4mm} \sum_{\substack{W(u) \in \mathbf{W}(u)\\  W(v)
  \in \mathbf{W}(v)}} \hspace{-1mm} \sum_{i=2}^\ell \Pr[W(u)] \cdot \Pr[W(v)]\cdot
               I (W(u), W(v), i) \\
&= \sum_{i=2} \Pr[W(u) \textrm{ and } W(v) \textrm{ first meet at } i]
  \\
&=\Pr[W(u) \textrm{ and } W(v) \textrm{ meet}]
\end{align*}
Note that $s(u, v)$ is the probability that $W(u)$ and $W(v)$ meet,
and the Lemma follows.
\end{proof}



By Lemma~\ref{lem:unbiasness} and Chernoff bound, we have the following Theorem that states by performing $n_r
= {3c \over \varepsilon} \log {n\over \delta}$ independent trials, the
error of the estimator
$\s(u,v)$ provided by Algorithm~\ref{alg:single_source} can be bounded
with high probability.

\begin{theorem}
\label{thm:single-source}
For every node $v \in V$, $v \neq u$, Algorithm
\ref{alg:single_source} returns an estimation $\s(u,v)$ for $s(u,v)$
such that
$\Pr[ \forall v \in V, |\s(u,v) - s(u,v)| \le \varepsilon] \ge 1-
\delta. $
\label{thm:EScoreAbsError}
\end{theorem}

We need the following form of Chernoff bound:
\begin{lemma}[Chernoff Bound \cite{ChungL06}] \label{lmm:chernoff}
For any set $\{x_i\}$ ($i \in [1, n_x]$) of i.i.d.\ random variables with mean $\mu$ and $x_i \in [0, 1]$,
$$\Pr\left\{\left|\sum_{i=1}^{n_x} x_i - n_x \mu\right| \geq n_x \e\right\} \leq \exp\left(-\dfrac{n_x \cdot \e^2}{\frac{2}{3}\e + 2\mu}\right).$$
\end{lemma}

\begin{proof}[of Theorem~\ref{thm:single-source}]
We first note that in each trial $k$, the estimator $\s_k(u, v)$ is a
value in $[0,1]$. It is obvious that $\s_k(u, v) \ge 0$. To see that
$\s_k(u, v) \le 1$, notice that $\s_k(u, v) = \sum_{i=2}^\ell
P(v, W(u, i))$ is a probability. More precisely, it is  the probability that a $\sqrt{c}$-walk $W(v)$ meets
with $\sqrt{c}$-walk $W(u)$ using the same steps.

Thus, the final estimator $\s(u,v) = {1 \over n_r}\sum_{k=1}^{n_r}
\s_k(u, v)$ is the average of $n_r$ i.i.d. random variables whos
values lie in the range $[0,1]$. Thus, we can apply Chernoff bound:
$$\Pr[|\s(u,v) - s(u,v)| \ge \varepsilon] \le \exp(-\varepsilon^2
n_r/(3 s(u,v))).$$
Recall that  $n_r =
\frac{3c}{{\varepsilon}^2}\mathrm{log}\frac{n}{\delta}$, and notice
that   $s(u ,v) \le c$, it follows that
$$\Pr[|\s(u,v) - s(u,v)| \ge \varepsilon] \le \exp\left(-\log {n \over
  \delta}\right) = {\delta \over n}.$$
 Taking union
bound over all nodes $v \in V$ follows that
$$\Pr[ \forall v \in V, |\s(u,v) - s(u,v)| \ge \varepsilon] \le
\delta,$$
and the Theorem follows.
\end{proof}


\section{Optimizations} \label{sec:single-opt}
We present three different optimization techniques to speed up our
basic {\em ProbeSim} algorithm. The {\em pruning rules} eliminates unnecessary traversals in
the \textsf{PROBE} algorithm, so that a single trial can be performed
more efficiently. The {\em batch} algorithm builds a reachability tree
to maintain all $n_r$ $c$-walks, such that we do not have to perform duplicated \textsf{PROBE}
operations in multiple trials. The randomized \textsf{PROBE} algorithm
reduces the worst-case time complexity of our algorithm to $O({n \over
  \varepsilon^2} \log {n \over \delta})$ in expectation.

\subsection{Pruning} \label{sec:single-opt-prune}
Although the expected steps  of a $\sqrt{c}$-walks is
$O(1)$, we may still find some long walks during a
large number of trials.
To avoid this overhead, we add the following pruning rule:

\begin{pruning}
Let $\varepsilon_t$
be the termination parameter to be determined later. In
Algorithm~\ref{alg:single_source}, truncate all $\sqrt{c}$-walks at
step $\ell_t = \log \varepsilon_t / \log \sqrt{c} $.
\end{pruning}

We explain the intuition of this pruning rule as follows. Let $W(u) = (u_1, \ldots, u_\ell)$ denote
the $\sqrt{c}$ walk, and $u_i$ denote a node on the walk with
$i>\ell_t$. For each node $v \in V$, $v \neq u$, the probability
that a $\sqrt{c}$ walk $W(v)$ meets $W(u)$ at $u_i =v_i$ is at most
$(\sqrt{c})^i = (\sqrt{c})^{i- \ell_t-1} \cdot  (\sqrt{c})^{\ell_t+1} \le (\sqrt{c})^{i - \ell_t-1} \varepsilon_t$, which means
that $u_i$ will contribute at most $(\sqrt{c})^{i -
  \ell_t-1}\varepsilon_t$ to the SimRank $s(u, v)$. Summing up over
$i = \ell_t+1, \ldots, \ell$ results in an error of ${1 \over
1-\sqrt{c}} \cdot \varepsilon_t$.  As we shall see in Theorem~\ref{thm:pruning},  more elaborated
analysis would show that the error contributed by this pruning rule is
in fact bounded by $\varepsilon_t$. We further notice that it is one-sided error,
we can add $\varepsilon_t/2$ to each estimator, which will reduce the pruning error by a factor
of 2.

The next pruning rule is inspired by the fact that the \textsf{PROBE} algorithm may traverse many
vertice with small scores, which can be ignored for the
estimation:

\begin{pruning}
Let $\varepsilon_p
$ denote the
pruning parameter to be determined later. In Algorithm \ref{alg:Probe},
after computing all $(v, Score(v))$ in $\mathcal{H}_{j}$ and before descending
to $\mathcal{H}_{j+1}$, we remove $(x,
Score(x))$ from $\mathcal{H}_{j}$  if
$Score(x) \cdot (\sqrt{c})^{i-j-1} \leq \varepsilon_p$.
\end{pruning}

The intuition of pruning rule 2 is that after $i-j-1$ more iterations
in Algorithm~\ref{alg:Probe}, the
scores computed from $Score(x, j)$ will drop down to $Score(x) \cdot
{\sqrt{c}}^{i-j-1} \le \varepsilon_p$. One might think that a node $v$
may get multiple error contributions from different pruned nodes;
However, the key insight is that the probabilities of the walks from $v$
to these nodes sum up to at most $1$, which implies that the error
introduced by a single probe is at most $\varepsilon_p$. We will make
this argument rigorous in Theorem~\ref{thm:pruning}.

\vspace{2mm}
\noindent{\bf Running Example for the Pruning Rules.}
Consider $\sqrt{c}$-walk $W(a)=(a,b,a,b,e)$, and set the
termination and pruning parameters to be $\varepsilon_t
=\varepsilon_p=0.05$. We first note that the length of $W(a)$ is $\ell
=5$, and $({\sqrt{c}})^{\ell_t} < 0.05$, so we truncate $W(a)$ to
$(a,b,a,b)$.

Now consider the \textsf{PROBE} algorithm on $W(a,4) = (a,b,a,b)$. In
the second iteration in  Figure~\ref{fig:probe}, recall that we have
$Score(c, 1) =0.167$. There are still two more iterations to go,
and yet we have $Score(c, 1) \cdot (\sqrt{c})^2  = 0.042 <
\varepsilon_p$. Thus pruning rule 2 takes place, and the algorithm
does not  have to descend to the subtree of $c$.

\vspace{2mm}
\noindent{\bf Correctness.} For ease of presentation, we refer to the
sampling error parameter in algorithm~\ref{alg:single_source}
as $\varepsilon$ and the maximum allowed error as $\varepsilon_a$.  Recall that $\varepsilon_t$ and $\varepsilon_p$ denote the termination parameter and
pruning parameter, respectively. The following Theorem
shows how these
parameters affect the final error. Essentially, the error introduced
by the two  pruning rules
is roughly the same as the sampling error. The proof of the Theorem can be found in
the full version~\cite{fullversion} of the paper.
\begin{theorem}
\label{thm:pruning}
Assume $\varepsilon$,  $\varepsilon_t$ and $\varepsilon_p$ satisfies the
following inequality
$\varepsilon + {1 + \varepsilon \over 1-\sqrt{c}} \cdot
\varepsilon_p +  {1 \over 2} \cdot \varepsilon_t \le \varepsilon_a, $
then Algorithm~\ref{alg:single_source} achieves
$\Pr[ \forall v \in V, |\s(u,v) - s(u,v)| \le \varepsilon_a] \ge 1-
\delta. $
\end{theorem}

\clearpage
\subsection{Batching Up $\sqrt{c}$-walks} \label{sec:single-opt-batch}
A simple observation is that if two  $\sqrt{c}$-walks $W_1(u)$ and $W_2(u)$ share the same
partial walk $W(u,i)=(u_1, u_2, \ldots, u_i)$, then we can perform a single
probe on $W(u,i)$ to return the scores for both
$\sqrt{c}$-walks. Meanwhile, we expect many $\sqrt{c}$-walks to share
a common partial walk when the number of trial $n_r$ is large. Based on this
observation, we propose to batch up $\sqrt{c}$-walks before we perform
the \textsf{PROBE} algorithm.

More precisely, we
use a {\em reverse reachability tree} $T$ to compactly store all $n_r$
$\sqrt{c}$-walks. Each tree node $r_q$ in $T$ maintains a graph node
$r_q.node$ and a weight integer $r_q.weight$. Let $(r = r_1, \ldots,
r_q)$ denote the tree
path from root $r$ to $r_q$, then $(r = r_1.node, \ldots,
r_q.node)$ corresponds to a partial $\sqrt{c}$-walk that starts at
$u$, and $r_q.weight$ is set to be the
number of $\sqrt{c}$-walks that shares this partial walk. In particular, the root $r$ maintains
$r.node =u$ and $r.weight=n_r$.


\begin{algorithm}[t]
\begin{small}
\caption{ Batch algorithm\label{alg:batch}}
\KwIn{Directed graph $G=(V,E)$; $u \in V$; Error $\varepsilon_a$ and failure probability $\delta$; \\}
\KwOut{$\mathcal{R}=\{(v, \s(u,v)) \mid v \in V\}$, a hash\_set of
  size $n$ that maintains the SimRank estimations for
  each node $v \in V$}
$n_r \leftarrow \frac{3c}{{\varepsilon}^2}\mathrm{log}\frac{n}{\delta}$; \\
Initialize reverse reachability tree $T$ rooted by $r$, with $r.node =u$ and
$r.weight =0$; \\
\For{ $k=1$ to $n_r$}
{
Generate $\sqrt{c}$-walk $W_k(u)=(u_1,u_2,...,u_l)$; \\
$r_1 \leftarrow r$;\\
\For {$i=2$ to $\ell$}
{
\If{$\exists r_i \in r_{i-1}.children$ and $r_i.node = u_{i}$}
{
$r_i.weight \leftarrow r_i.weight+ 1$;\\
}
\Else
{
Add $r_i$ as a child to $r_{i-1}$, with $r_i.node = u_i$ and
$r_i.weight =1$;\\
}
}
}
\For {each root-to-node path $(r=r_1, r_2, \ldots, r_q)$ in $T$ }
{
$\mathcal{S} \leftarrow \mathsf{PROBE}((r_1.node, \ldots,
r_q.node))$; \\
\For {each $(v, Score(v)) \in \mathcal{S}$}
{
$\s(u,v)\leftarrow \s(u,v) + \frac{r_q.weight}{n_r} \cdot
  Score(v)$;
}
}
\Return $\mathcal{R}$;
\end{small}
\end{algorithm}

Algorithm~\ref{alg:batch} illustrates the pseudo-code of the batch
method. The algorithm generates $n_r =
\frac{3c}{{\varepsilon}^2}\mathrm{log}\frac{n}{\delta}$ random
$\sqrt{c}$-walks, where $\varepsilon$ is a constant that satisfies
Theorem~\ref{thm:pruning}. To insert a $\sqrt{c}$-walk
$W_i(u)=(u_1,u_2,...,u_\ell)$ to $T$, the algorithm starts from the root $r_1= r$,
and recursively move down to the lower level. After $r_{i-1}$ is
processed,  it checks if there is a child of $r_{i-1}$,
denoted $r_i$,  such that
$r_i.node = u_i$ (Lines 5-7). If so, we know that partial $\sqrt{c}$-walk
$(u_1,u_2,...,u_i)$ is already recorded by $r_1, \ldots, r_i$, so the algorithm
increases the weight of $r_i$ by $1$ (Lines 8-9). Otherwise, it adds a new child
$r_i$ to $r_{i-1}$,  with $r_i.node = u_i$ and $r_i.weight = 1$ (Lines
10-11). After all $n_r$ walks are inserted to $T$, the algorithm starts the probe processes.
For each root-to-node
path $(r=r_1, r_2, \ldots, r_q)$ in $T$, the algorithm invokes the \textsf{PROBE}
algorithm on $(r_1.node, \ldots, r_q.node)$ to get the score set
$\mathcal{S}$ (Lines 11-12). Note that we can apply BFS traversal to $T$
to enumerate all root-to-node paths. Since the number of $\sqrt{c}$-walks that share
$(r_1.node, \ldots, r_q.node)$ as their partial $\sqrt{c}$-walk is
$r_q.weight$, the algorithm adds  $ \frac{r_q.weight}{n_r} \cdot Score(v) $ to the
estimation of SimRank $\s(u, v)$, for each $(v, Score(v)) \in
\mathcal{S}$ (Lines 13-14). Finally, the algorithm returns
$\mathcal{R}$ as the SimRank estimators (Line 15).

\vspace{2mm}
\noindent{\bf A running example for the batch algorithm.}
Suppose we have a reverse reachability tree $T$
shown
in~\ref{fig:reachability_tree_before},
which records two $\sqrt{c}$-walks
$(a, b, c)$ and $(a, c, a)$. To
insert $W=(a,b,a)$ to $T$, the
algorithm starts with the root $r_1$
and increase $r_1.weight$ by 1. Then
it finds the child $r_2$ such that
$r_2. node = b$, and increase
$r_2. weight$ by 1. Finally, since
there is no child node of $r_2$ with
node equal to $a$, the algorithm inserts a new
node $r_6$ as a child node of $r_6$,
with $r_6.node =a$, and
$r_6.weight =1$.

To estimate the SimRank values using
the reverse reachability tree $T$, we probe
partial walk represented by each tree node and sum
up the scores according to the
weights of nodes. For example, let
$Score(b, r_j)$ denote the score
computed by the \textsf{PROBE}
algorithm for node $b$ on the partial walk represented by $r_j$, for $j=2, 3, 4,
5,6$. Since the
weights of $r_2, r_3, r_4, r_5, r_6$ are
$2,1,1,1, 1$, the final estimator
$\s(a, b)$ will be set to ${1 \over 3} \cdot
Score(b, r_2)+ {1 \over 6} \cdot
Score(b, r_3) +{1 \over 6}  \cdot
Score(b, r_4) + {1 \over 6} \cdot
Score(b, r_5) +{1 \over 6} \cdot  Score(b, r_6)$.

\begin{figure}[t]
\label{fig:reachtree}
\begin{centering}
\vspace{-4mm}
\subfigure[]{\begin{centering}
\includegraphics[scale=0.17]{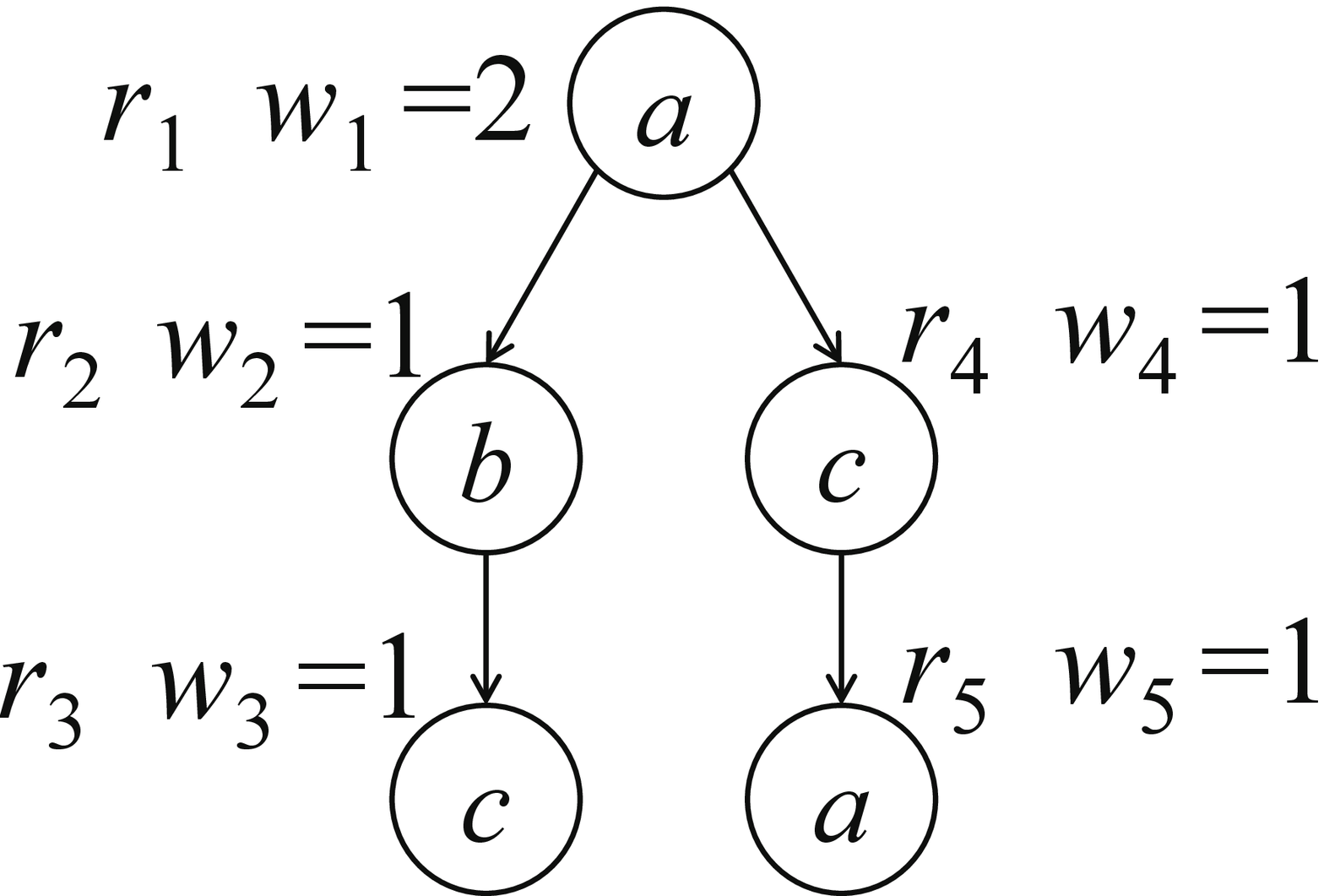}
\label{fig:reachability_tree_before}
\par\end{centering}
}
\subfigure[]{\begin{centering}
\includegraphics[scale=0.17]{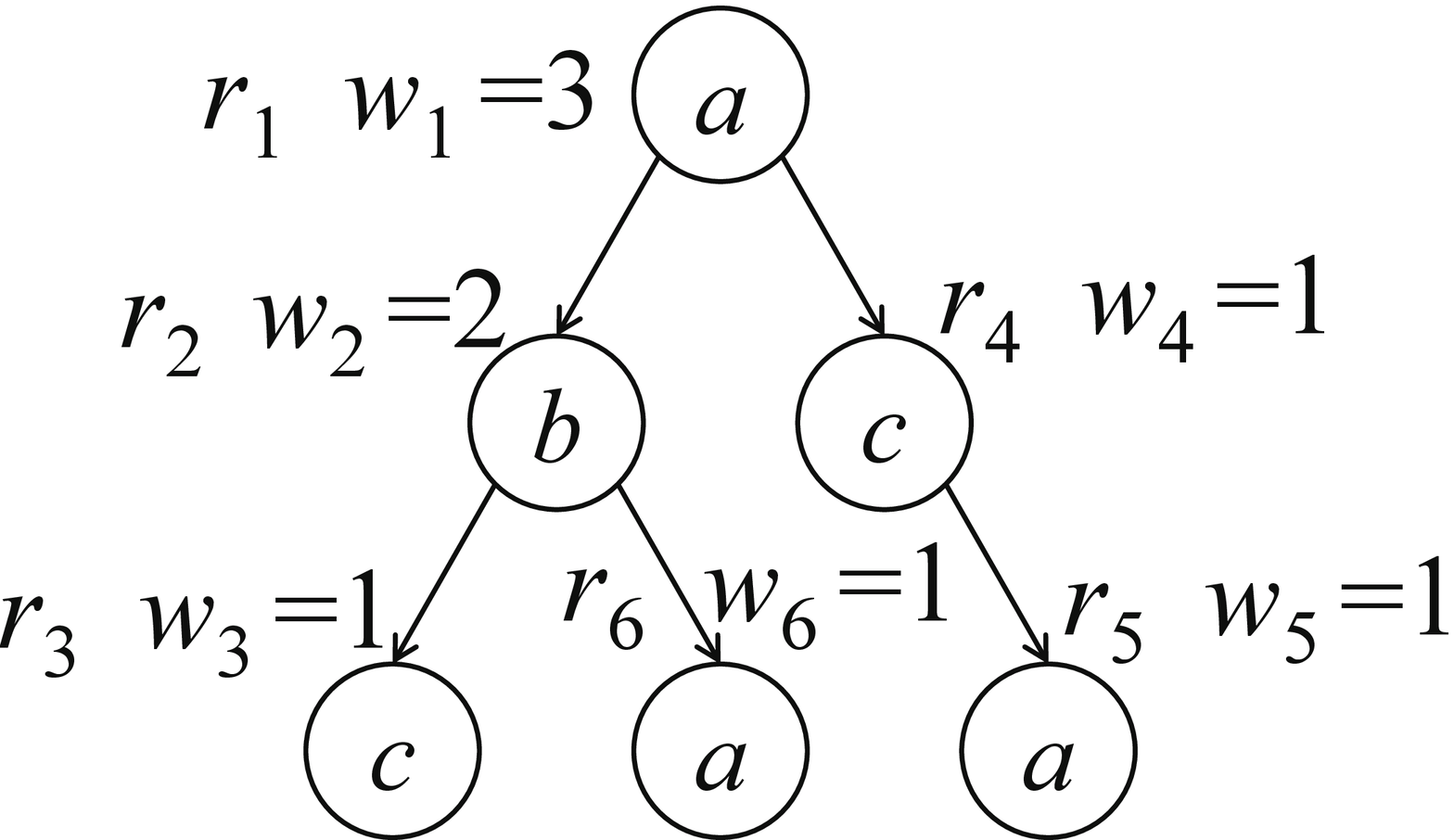}
\label{fig:reachability_tree_after}
\par\end{centering}
}
\par\end{centering}
\caption{An example of reverse reachability tree}
\end{figure}

\subsection{Randomized \textsf{PROBE} Algorithm} \label{sec:single-opt-probe}
Recall that the running time of Algorithm~\ref{alg:single_source} is $O({m \over
  \varepsilon_a^2} \log {n \over \delta})$. The $m$ factor comes from
the \textsf{PROBE} algorithm, which runs in $O(m)$ time.
To overcome this worst-case complexity, we present a randomized
version of the \textsf{PROBE} algorithm. This algorithm
 runs in $O(n)$ time in expectation. The intuition
is that for each iteration,  instead of deterministically probing each out-neighbours of
 the nodes and computing corresponding scores,  we simply sample the
 in-neighbours of EVERY nodes in the graph to determine if it
 should be put into the next iteration. We delicate the
 sampling process such that the probability that $v$
 gets selected by $j$-th iteration is exactly the score of $v$ computed by the
 deterministic \textsf{PROBE} algorithm at $j$-th iteration. Since
 an iteration touches each node $v\in V$ exactly once, and there are
 constant number of iterations in expectation, the
 expected running time is bounded by $O(n)$.

Algorithm~\ref{alg:RProbe} shows the pseudo-code of the randomized
\textsf{PROBE} algorithm. Given a partial $\sqrt{c}$-walk $W(u, i) =(u_1, \ldots, u_i)$ that
starts at $u=u_1$, the randomized \textsf{PROBE} algorithm
outputs $\mathcal{S}=\{(v, Score(v)) \mid  v \neq u \in V\}$, a hash\_set of
nodes and their first-meeting probability with respect to reverse
path $W(u, i)$.
Similar to its deterministic sibling, the algorithm initializes $i-1$ hash
  tables $\mathcal{H}_0, \ldots, \mathcal{H}_{i-1}$ (Line 1), and adds
  $u_i$ to $\mathcal{H}_0$ (Line 2). In the
  $j$-th iteration, the algorithm first checks if the sum of
  out-degrees of the nodes in $\mathcal{H}_j$ is below $n$ (Lines 3-4). If so, the
  algorithm sets the candidate set $U$ to the union of the
  out-neighbours (Line 5); otherwise, it simply sets $U$ as $V$
  (Lines 6-7). After that, the algorithm samples each node $v \in U$ with the following
procedure. First, it uniformly selects an incoming edge $(x, v)$ from the
in-neighbour set $I(v)$ of $v$ (Line 9). If $x \in \mathcal{H}_{j}$ (i.e., $x$ is
selected in iteration $j-1$), the algorithm selects $v$ into $\mathcal{H}_{j+1}$ with probability $\sqrt{c}$ (Lines 10-11).
 After $i$ iterations, the algorithm returns all nodes in $\mathcal{H}_{i-1}$ with their scores set to be $1$ (Line 12).


\begin{algorithm}[t]
\begin{small}
\caption{Randomized $\mathsf{PROBE}$ algorithm\label{alg:RProbe}}
\KwIn{A partial $\sqrt{c}$-walk $(u=u_1, \ldots,
  u_i)$}

\KwOut{$\mathcal{S}=\{(v, Score(v)) \mid  v \in V\}$, a hash\_set of
  nodes and their scores w.r.t. partial walk $W(u,i)$}
Initialize hash\_set $\mathcal{H}_j$ for $j=0,\ldots, i-1$;\\
Insert $(u_i,1)$ to $\mathcal{H}_0$; \\
\For{ $j=0$ to $i-2$}
{
\If{$\sum_{v
  \in \mathcal{H}_j} |\mathcal{O}(v)| \le n$}
{$U \leftarrow \bigcup_{v \in \mathcal{H}_j} \mathcal{O}(v)$;}
\Else
{$U \leftarrow V$;}
\For{each $x \in U$, $x \neq u_{i-j-1}$}
{
Uniformly sample an edge $(v, x)$ from $I(x)$;\\
\If {$v \in \mathcal{H}_j$}
{
Insert $x$ to $\mathcal{H}_{j+1}$ with probability $\sqrt{c}$;
}
}
}
\Return $\mathcal{S} = \{(v, 1) \mid  v \in \mathcal{H}_{i-1}\}$\;
\end{small}
\end{algorithm}



\vspace{2mm}
\noindent{\bf Time complexity.} Each iteration in
Algorithm~\ref{alg:RProbe} runs $O(n)$ time, so the randomized \textsf{PROBE}
algorithm for a partial $\sqrt{c}$-walk $W(u,i)$ of length $i-1$
runs in $O(i\cdot n)$
time. If we use randomized \textsf{PROBE} in
Algorithm~\ref{alg:single_source}, the running time for a single walk is bounded by
$O(\sum_{i=1}^\ell i\cdot n) = O(\ell^2 n).$
As proved in Section~\ref{sec:analysis_basic} the expectation of $\ell^2$ is a
constant, thus the expected running for a single walk with randomized
\textsf{PROBE} is bounded by $O(n)$.
By setting the number of $\sqrt{c}$-walks to be $n_r =O(
\frac{1}{{\varepsilon}^2}\log\frac{n}{\delta})=O(
\frac{1}{{\varepsilon_a}^2}\log\frac{n}{\delta})$,
the time complexity of our single source SimRank
algorithm with randomized \textsf{PROBE} is at bounded by
$O(\frac{n}{{\varepsilon_a}^2} \log
\frac{n}{\delta})$. We also note that in practice, the randomized  \textsf{PROBE}
algorithm tends to only visit the nodes that can be reached by $u_i$
with non-negligible probabilities, which is few in number for
real-world graphs that follow the power-law distribution.

\vspace{2mm}
\noindent{\bf Correctness.}
The following Theorem shows that the randomized \textsf{PROBE} algorithm gives an
unbiased Bernoulli estimators for the scores computed by the deterministic \textsf{PROBE}
algorithm.
The proof of the Theorem can be found in
the full version~\cite{fullversion} of the paper.

\vspace{-2mm}
\begin{theorem}
\label{thm:randomized_error}
For EVERY node $v \in V$, $v\neq u$, Algorithm~\ref{alg:single_source} with randomized
\textsf{PROBE} returns an estimation $\s(u,v)$ for $s(u,v)$
such that
$\Pr[ \forall v \in V, |\s(u,v) - s(u,v)| \le \varepsilon_a] \ge 1-
\delta. $
\end{theorem}

\subsection{Best of both worlds}
\label{sec:best}
Although the randomized \textsf{PROBE} algorithm achieves better
worst-case time complexity, it still suffers from one drawback: the
sampling processes cannot be batched up.
Recall that in the batch algorithm, each partial
$\sqrt{c}$-walk $W(u,i)$ with weight $w$ is probed with the deterministic \textsf{PROBE}
algorithm exactly once, regardless of what $w$ is. If we
switch to the randomized \textsf{PROBE} algorithm, however, we have to
perform $w$ independent probes and take the average to get an unbiased
estimator for each node. This means that batching up the $\sqrt{c}$-walks
does not reduce the running time of our single source SimRank
algorithm, if we use the randomized \textsf{PROBE} algorithm.

Now we have a deterministic
\textsf{PROBE} algorithm that can be batched up, and a randomized
\textsf{PROBE} algorithm that achieves $O(n)$ time complexity. To get
the best of both world,
we can combine the two algorithms to cope with the batch
algorithm. The idea is very simple: Let $r_q$ be a tree node  with
weight $w$, and consider a probe for the partial $\sqrt{c}$-walk for
$r_1.node, \ldots, r_q.node$. After each iteration (say $j$-th) in the
deterministic \textsf{PROBE} algorithm, we check
if the summation of the out-degrees exceeds $c_0 w n
$ for some constant $c_0$.
If so, we know that the deterministic \textsf{PROBE} algorithm is going
to incur a time complexity of at least $c_0w n$, and
is no longer suitable for this partial path. Thus,  we will  switch to the
randomized \textsf{PROBE} algorithm,  which runs in $O(w n)$ time for a single probe.

The intuition of this combination can be explained as follow. If the
partial $\sqrt{c}$-walk $W(u)= (u_1, \ldots, u_i)$ is short, then it
is likely that the weight $w$ of this partial walk is large. In the mean time, the number of
nodes that are  can reversely reach $u_i$ using $\le i-1$ steps is likely to be small,
so we can afford to use the deterministic \textsf{PROBE}  algorithm to calculate
the scores exactly and save a factor of $w$. On the other hand, if the
partial  $\sqrt{c}$-walk is very long, then it is likely that the
weight $w$ is small, and thus performing the randomized \textsf{PROBE}
algorithm independently $w$ times is affordable.

\section{Related Work} \label{sec:related}

The first method for SimRank computation \cite{JW02}, referred to as the {\em Power Method}, is designed for deriving the SimRank similarities of all node pairs in the input graph $G$. It utilizes the following matrix formulation of SimRank \cite{KMK14}:
\begin{equation} \label{eqn:related-simrank}
S = (c P^\top S P) \vee I,
\end{equation}
where $S$ is an $n\times n$ matrix such that $S(i, j)$ equals the SimRank similarities between the $i$-th and $j$-th nodes, $I$ is an $n \times n$ identity matrix, $c$ is the decay factor in the definition of SimRank, $P$ is a {\em transition matrix} defined by the edges in $G$, and $\vee$ is the element-wise maximum operator. The power method starts by setting $S=I$, and then it iteratively updates all elements in $S$ based on Equation~\ref{eqn:related-simrank}, until the values of all elements converge. This method is subsequently improved in \cite{LVGT10,YZL12,YuJulie15gauging} in terms of either efficiency or accuracy. However, all methods proposed in \cite{JW02,LVGT10,YZL12,YuJulie15gauging} incur $O(n^2)$ space overheads, which is prohibitively expensive for large graphs.

To mitigate the inefficiency of the power method, a line of research \cite{FNSO13,He10,Yu13,Li10,Yu14,YuM15b,KMK14} has proposed to utilize an alternative formulation of SimRank that makes it easier to compute:
\begin{equation} \label{eqn:related-simrank-wrong}
S = c P^\top S P +  (1-c)\cdot I.
\end{equation}
This formulation is claimed to be equivalent to that in Equation~\ref{eqn:related-simrank} \cite{FNSO13,He10,Yu13,Li10,Yu14,YuM15b}. As pointed out in \cite{KMK14}, however, the two formulations are rather different, due to which the techniques in \cite{FNSO13,He10,Yu13,Li10,Yu14,YuM15b} do not always return the correct SimRank similarities of node pairs.

Among the existing solutions that adopt the correct formulation of SimRank (in Equations \ref{eqn:intro-simrank} and \ref{eqn:related-simrank}), the ones most related to ours are proposed in \cite{LeeLY12, TX16,FRCS05}, since they can answer approximate single-source and top-$k$ SimRank queries with non-trivial absolute error guarantees. In particular, Fogaras and R{\'{a}}cz \cite{FRCS05} propose a Monte Carlo approach that is similar to the method discussed in Section~\ref{sec:prelim-randomwalk}, except that it uses conventional random walks instead of $\scw$-walks. They also present an index structure that stores pre-computed random walks to accelerate query processing. As shown in \cite{KMK14,TX16}, however, the index structure incurs tremendous space and preprocessing overheads, which makes it inapplicable on sizable graphs. Lee et al.\ \cite{LeeLY12} propose an index-free algorithm for top-$k$ SimRank queries that is claimed to provide absolute error guarantees. Nevertheless, Lee et al.'s algorithm may return erroneous query results due to limited steps of random walks, as we discuss in Section~\ref{sec:intro} and ~\ref{sec:exp}. Tian and Xiao \cite{TX16} present {\em SLING}, an index structure that answers any single-source SimRank query in $O(m \log\frac{1}{\e_a})$ time and requires $O(n/\e_a)$ space. {\em SLING} is shown to outperform the state of the art in terms of both query efficiency and accuracy, but its space consumption is more than an order of magnitude larger the size of $G$. Furthermore, it cannot handle updates to the input graph, and its absolute error guarantee cannot be changed after the preprocessing procedure.

In addition, there exist two other index structures \cite{MKK14,SLX15} for top-$k$ SimRank queries, but neither of them is able to provide any worst-case error guarantee, since they reply on heuristic assumptions about $G$ that do not always hold. We discuss the technique in \cite{SLX15} in  Section~\ref{sec:intro}, and we refer interested readers to \cite{TX16} for explanations of the limitation of \cite{MKK14}.
Furthermore, Li et al.\ \cite{LiFL15} propose an distributed version of the Monte Carlo approach in \cite{FRCS05} and show that it can scale to a billion-node graph, albeit requiring 110 hours of preprocessing time and using 10 machines with 3.77TB total memory. Last but not least, there is existing work on {\em SimRank similarity join} \cite{TaoYL14,MKK15,ZhengZF0Z13} and variants of SimRank \cite{AMC08,FR05,Lin12,YuM15a,ZhaoHS09}, but the solutions therein cannot be applied to address top-$k$ and single-source SimRank queries.

\section{Experiments} \label{sec:exp}
This section experimentally evaluates the proposed solutions against
the state of the art. All experiments are conducted on a machine with
a Xeon(R) CPU E5-2620@2.10GHz  CPU and 96GB memory.
All algorithms are implemented in C++ and compiled by g++ 4.8.4 with the -O3 option.


\header
\noindent
{\bf Methods.} We evaluate six algorithms: {\em ProbeSim}, {\em MC} \cite{FR05},
{\em TSF} \cite{SLX15},  {\em TopSim} \cite{LeeLY12}, {\em Trun-TopSim}  \cite{LeeLY12} and {\em Prio-TopSim}  \cite{LeeLY12}.
As mentioned in Section~\ref{sec:related}, the three TopSim based algorithms
are the state-of-the-art index-free approaches for single-source and top-$k$ SimRank
queries, and {\em TSF} is the state-of-the-art index structure for
SimRank computations on dynamic graphs.

\header
{\bf Datasets.} We use $4$ small graphs and $4$ large graphs, as shown in
Table~\ref{tbl:datasets}. All datasets are obtained from public
sources \cite{SNAP,LWA}.

\begin{table}[t]
\centering
\tblcapup
\vspace{-3mm}
\caption{Datasets.}
\tblcapdown
\begin{small}
 \begin{tabular}{|l|l|r|r|} 
 \hline
 {\bf Dataset} & {\bf Type} & {\bf $\boldsymbol{n}$} & {\bf $\boldsymbol{m}$}	 \\ \hline
 Wiki-Vote	& 	directed &	7,155	&	103,689
   \\
 HepTh	    & 	undirected &	9,877	&	25,998		\\
AS	    &	directed &	26,475	&	106,762		\\
 HepPh	        &	directed &	34,546	&	421,578 \\
 LiveJournal &	directed	&	4,847,571	&	68,993,773		\\
 It-2004	&	directed & 41,291,594 & 1,150,725,436		\\
Twitter & directed & 41,652, 230 & 1 ,468, 365 ,182 \\
Friendster  & directed & 68,349,466 & 2,586,147,869 \\
 \hline
\end{tabular}
\end{small}
\label{tbl:datasets}
\vspace{-0mm}
\end{table}



\subsection{Experiments on Small Graphs} \label{sec:exp-results-small}

We first evaluate the algorithms on the four small graphs, where the ground
truth of the SimRank similarities can be obtained by the {\em Power
  Method}. On each dataset, we select $100$ nodes
uniformly at random from those with nonzero in-degrees. We
generate single-source and top-$k$ SimRank queries from each node to
evaluate the algorithms.


\begin{figure*}[!t]
\begin{small}
 \centering
    \begin{tabular}{cccc}
  \multicolumn{4}{c}{\hspace{-4mm} \includegraphics[height=4.5mm]{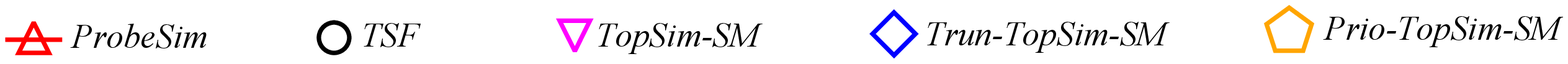}}
      \vspace{-1mm} \\
        \hspace{-8mm} \includegraphics[height=33mm]{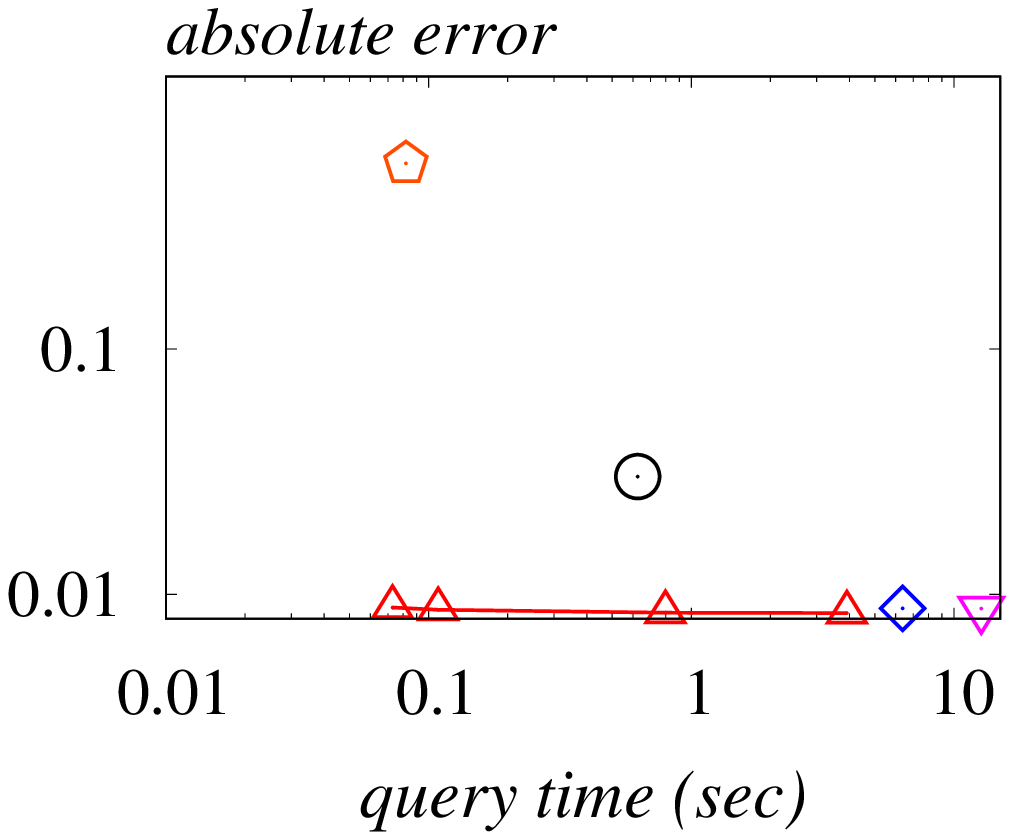}
&
        \hspace{-8mm} \includegraphics[height=33mm]{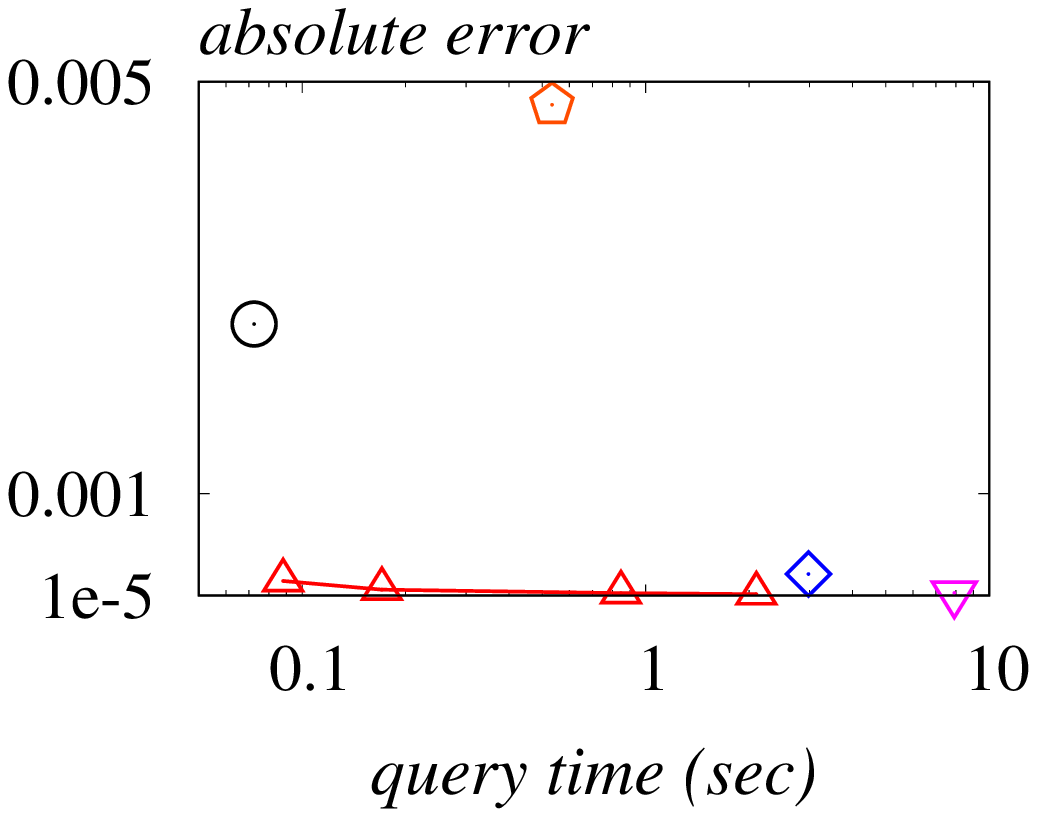}
&
        \hspace{-8mm} \includegraphics[height=33mm]{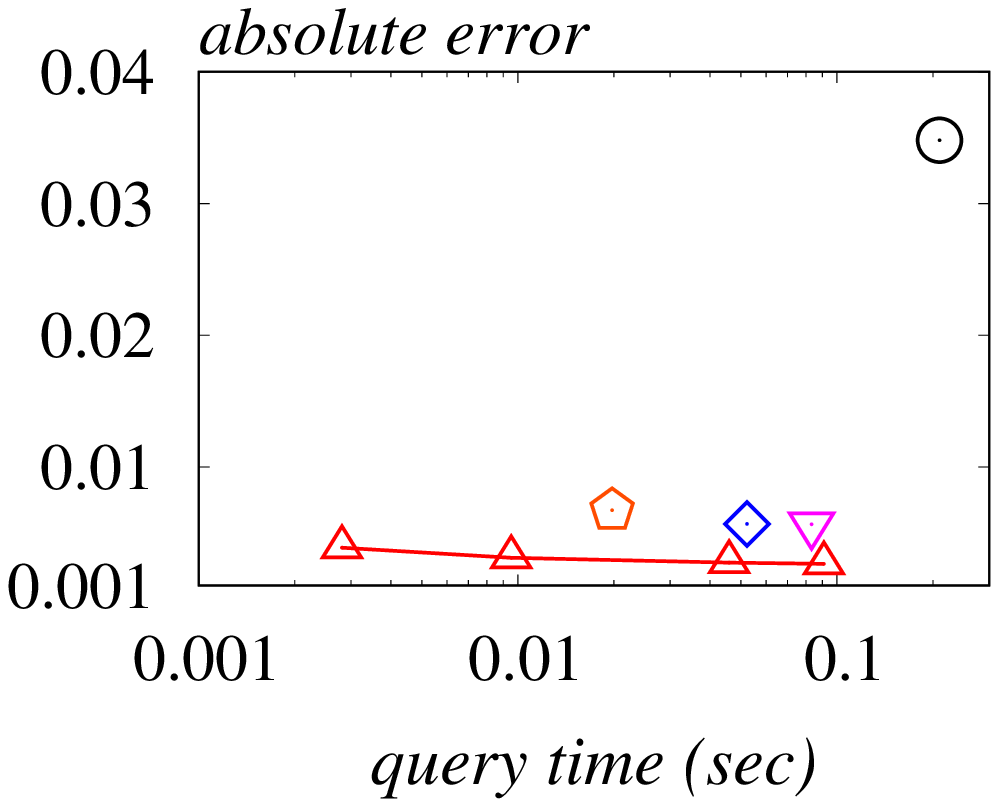}
&
        \hspace{-4mm} \includegraphics[height=33mm]{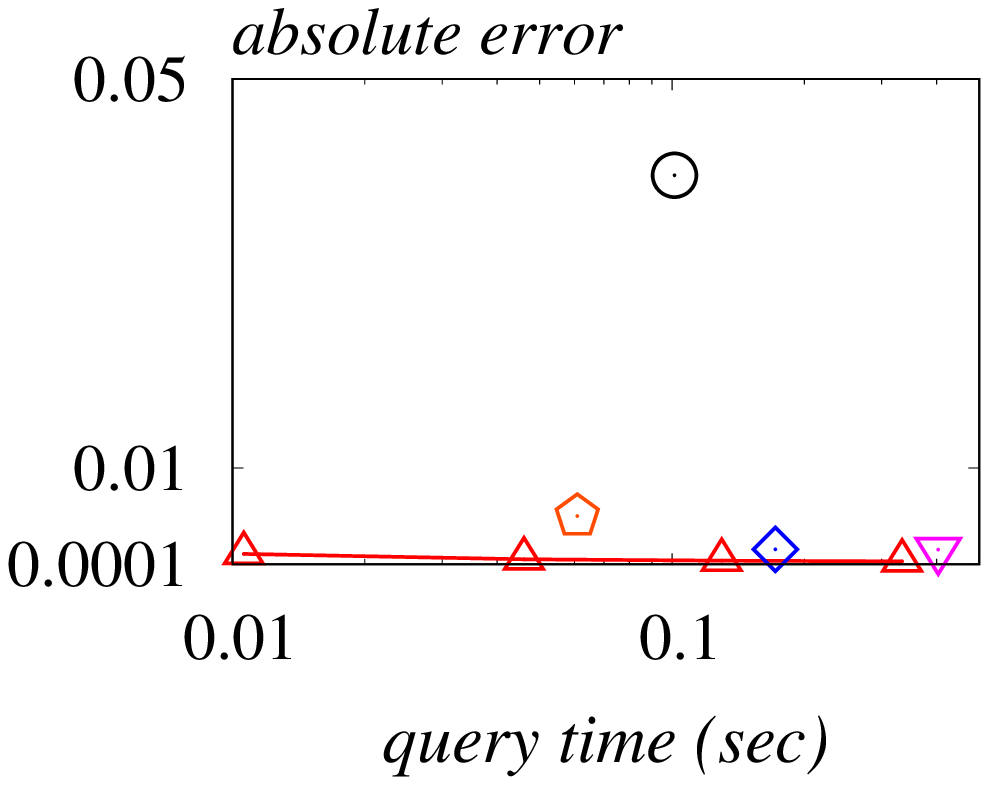}
   \vspace{-1mm} \\
       \hspace{-4mm} (a) {AS}  &
       \hspace{-8mm} (b) {Wiki-Vote} &
       \hspace{-8mm} (c) {HepTh} &
       \hspace{-8mm} (d) {HepPh} \\
 \end{tabular}
\vspace{-3mm}
 \caption{Absolute error in answering single-source SimRank queries on
 small graphs} \label{fig:exp-max-error}
\vspace{-0mm}
\end{small}
\end{figure*}

\begin{figure*}[!t]
\begin{small}
 \centering
    \begin{tabular}{cccc}
  \multicolumn{4}{c}{\hspace{-4mm} \includegraphics[height=4.5mm]{./Figs/legend_small.eps}}
      \vspace{-1mm} \\
        \hspace{-8mm} \includegraphics[height=33mm]{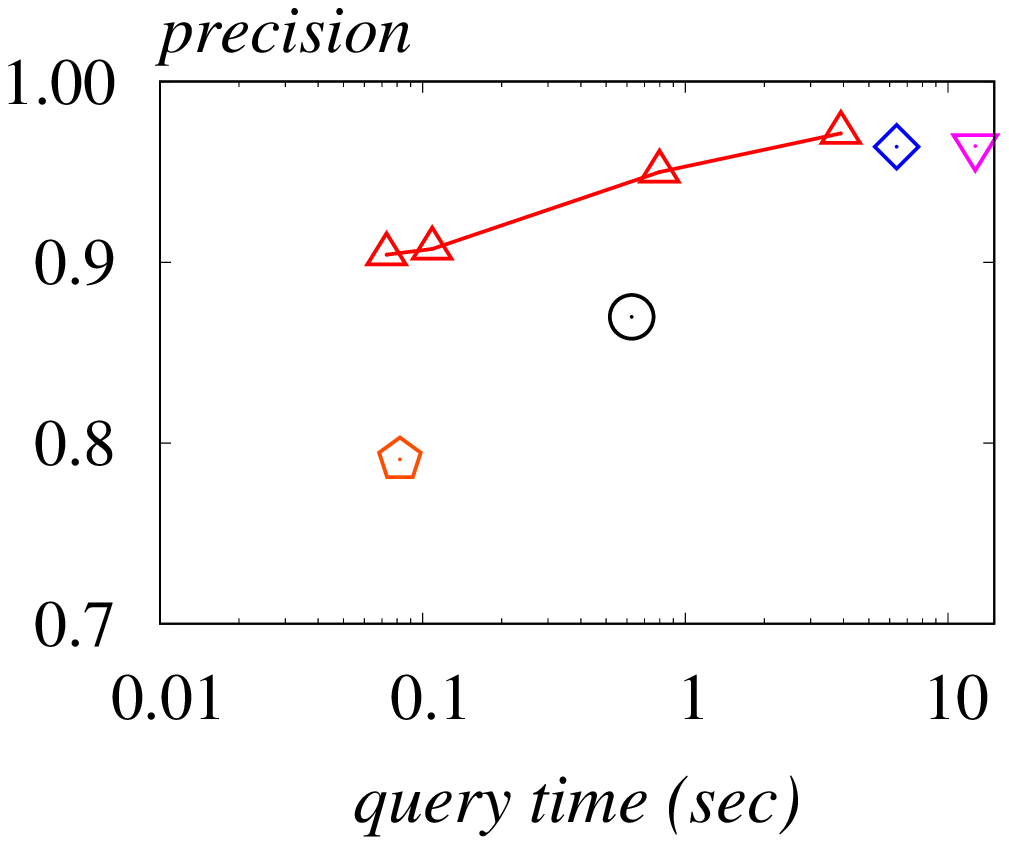}
   &
        \hspace{-8mm} \includegraphics[height=33mm]{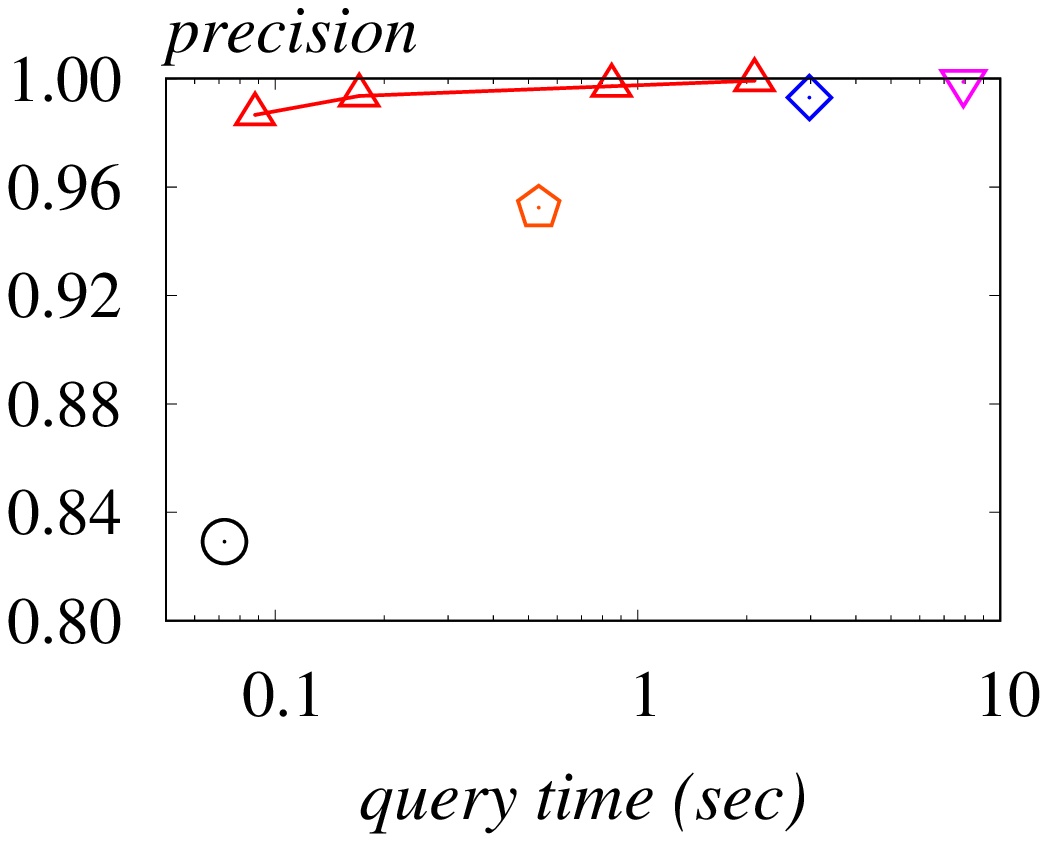}&
        \hspace{-8mm} \includegraphics[height=33mm]{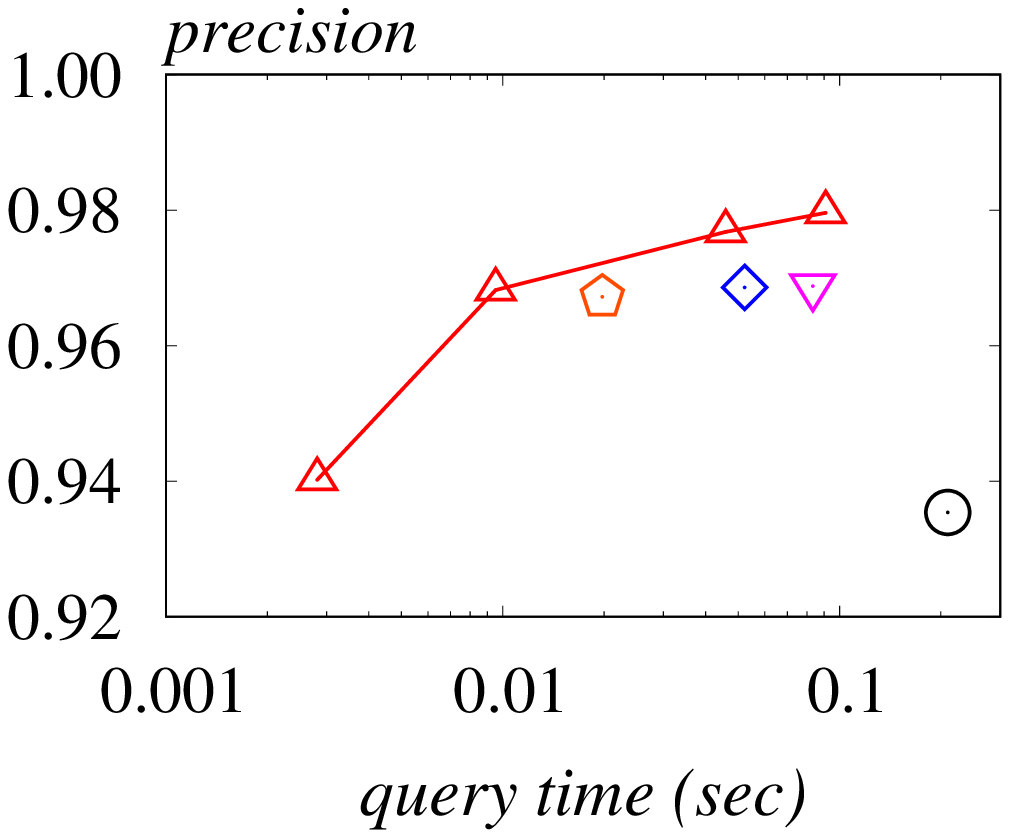} &
        \hspace{-4mm} \includegraphics[height=33mm]{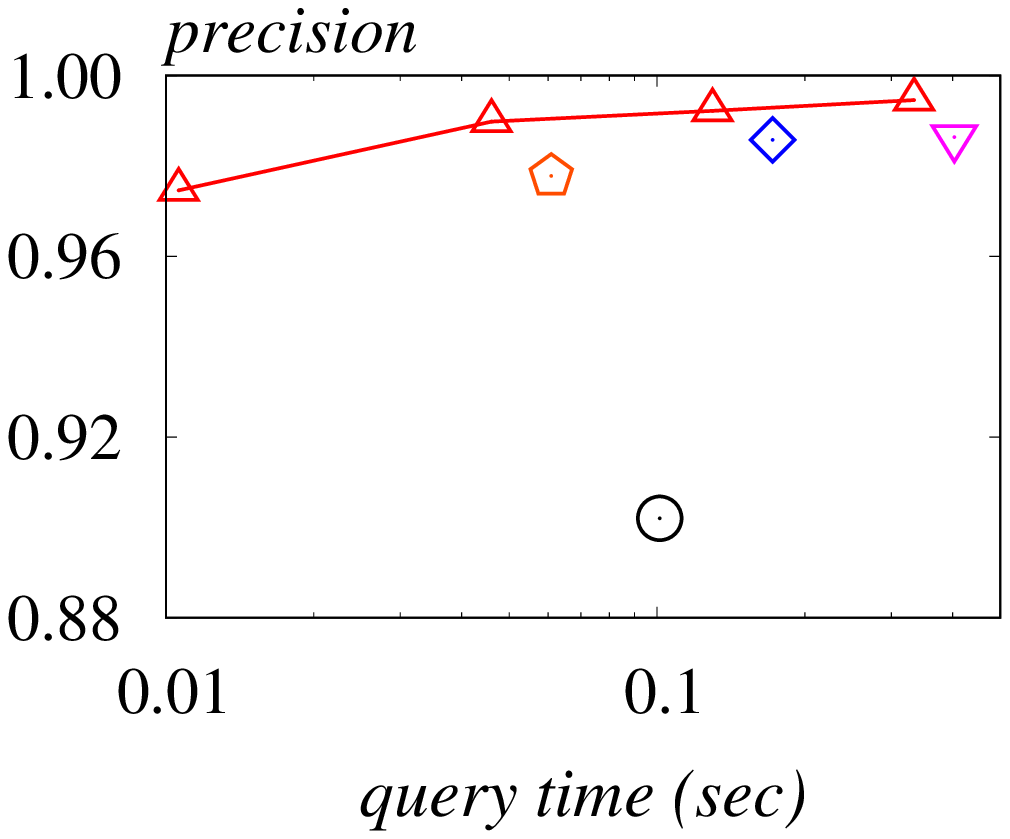}
   \vspace{-1mm} \\
       \hspace{-4mm} (a) {AS}  &
       \hspace{-8mm} (b) {Wiki-Vote} &
       \hspace{-8mm} (c) {HepTh} &
       \hspace{-8mm} (d) {HepPh} \\
 \end{tabular}
\vspace{-3mm}
 \caption{ {\em Precision@k} vs.\ query time for top-$\boldsymbol{k}$
   SimRank queries on small graphs} \label{fig:exp-precision-small}
\vspace{-0mm}
\end{small}
\end{figure*}

\begin{figure*}[!t]
\begin{small}
 \centering
    \begin{tabular}{cccc}
  \multicolumn{4}{c}{\hspace{-4mm} \includegraphics[height=4.5mm]{./Figs/legend_small.eps}}
      \vspace{-1mm} \\
        \hspace{-8mm} \includegraphics[height=33mm]{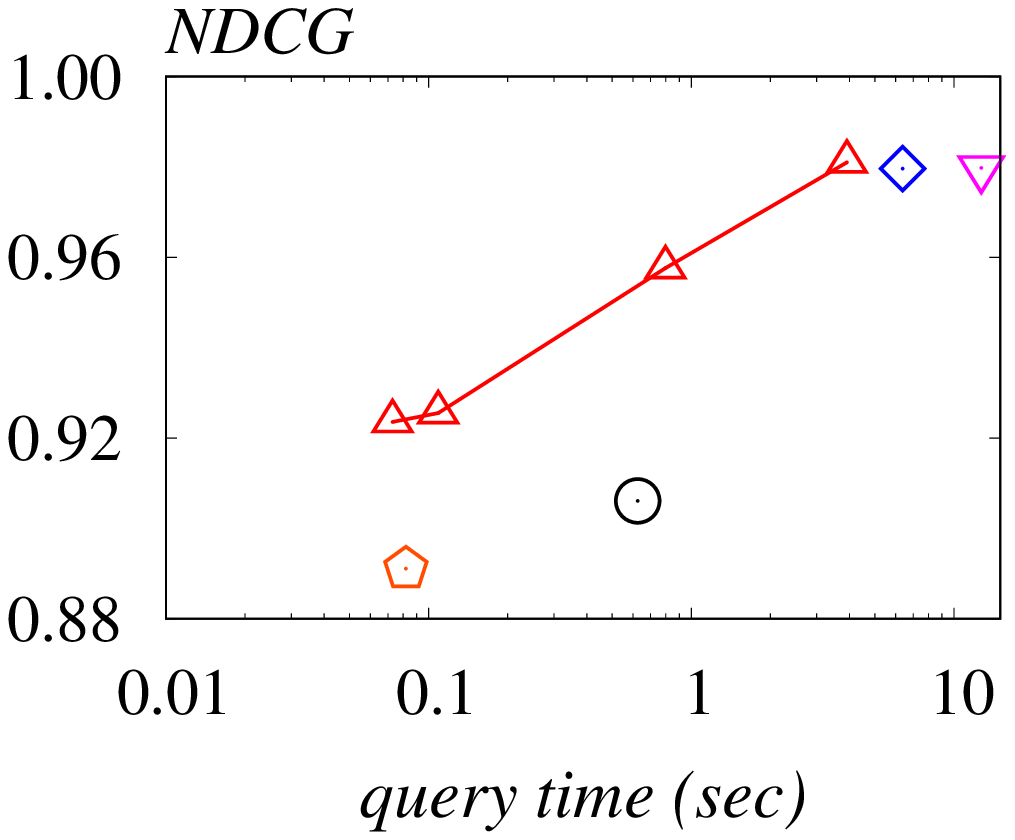}
   &
        \hspace{-8mm} \includegraphics[height=33mm]{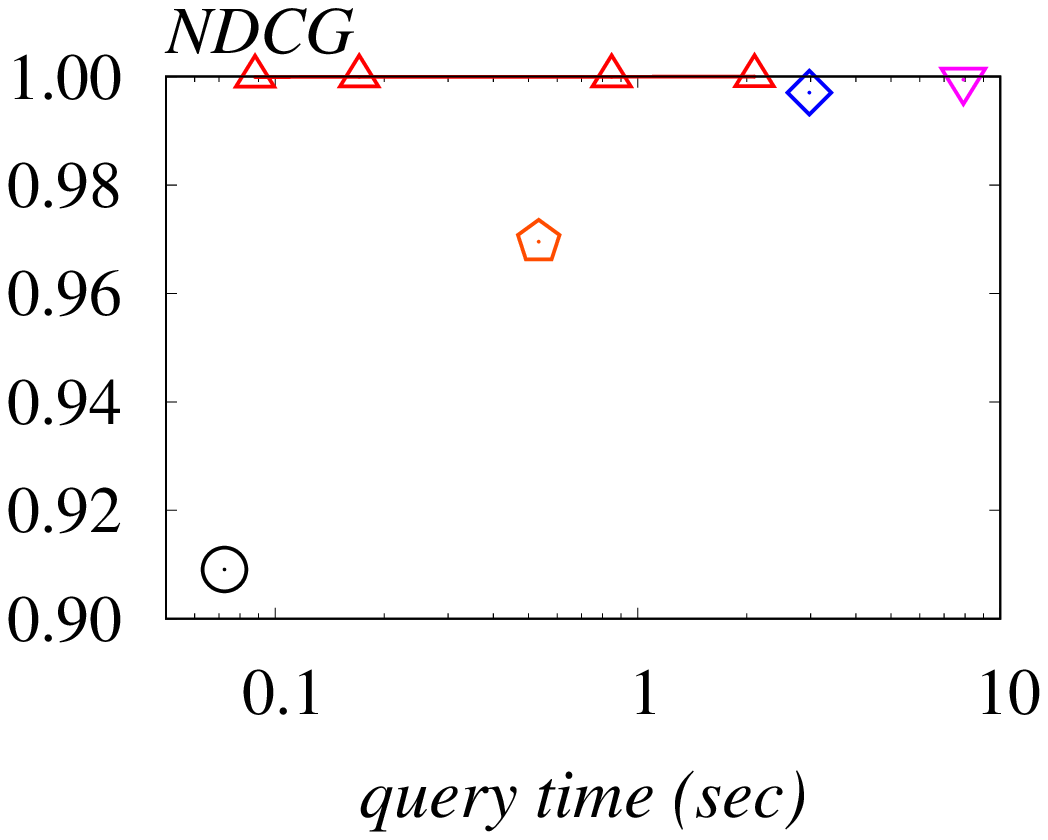}&
        \hspace{-8mm} \includegraphics[height=33mm]{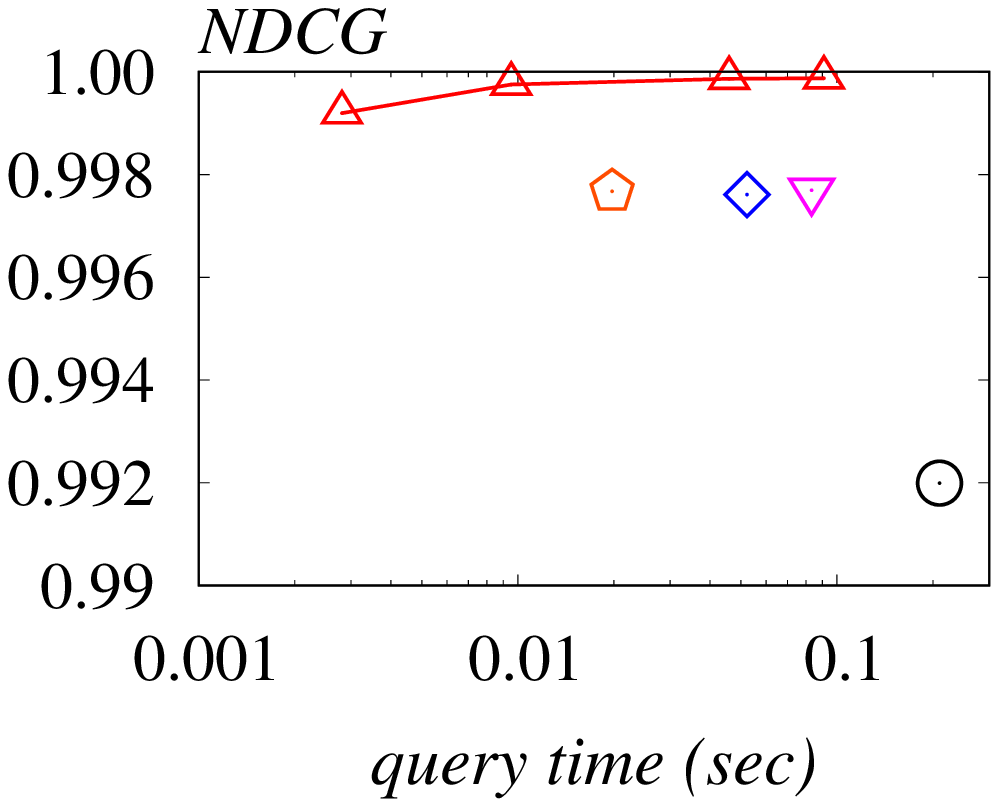} &
        \hspace{-4mm} \includegraphics[height=33mm]{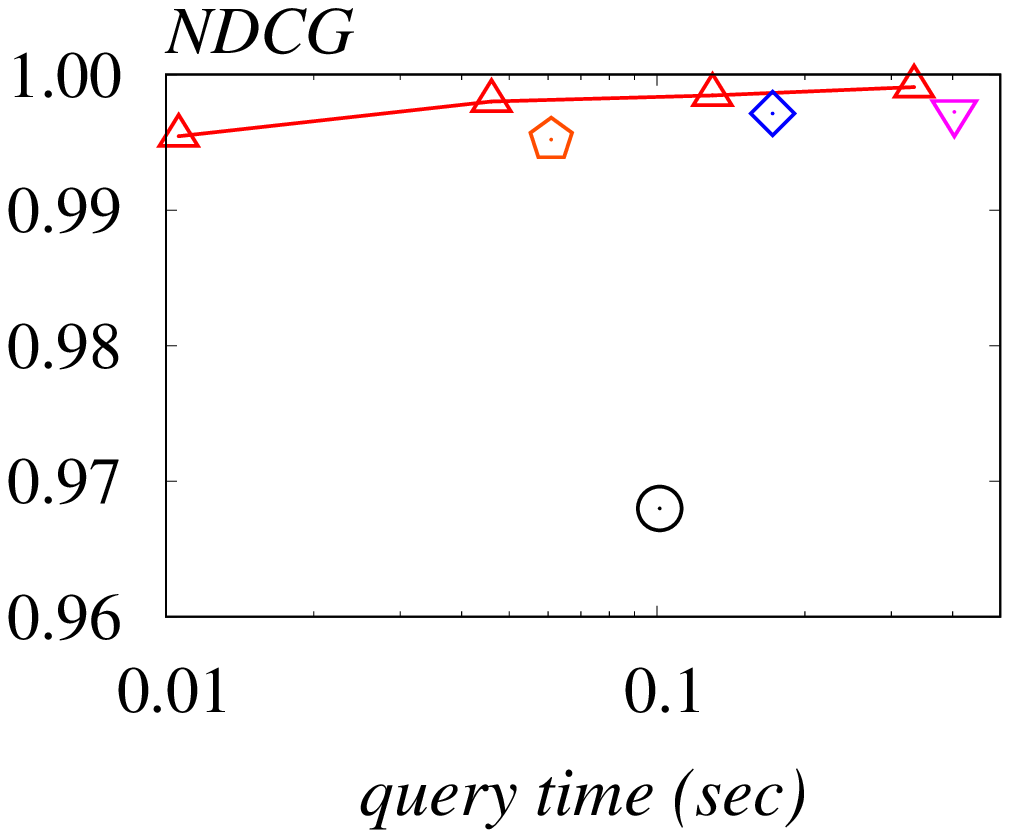}
   \vspace{-1mm} \\
       \hspace{-4mm} (a) {AS}  &
       \hspace{-8mm} (b) {Wiki-Vote} &
       \hspace{-8mm} (c) {HepTh} &
       \hspace{-8mm} (d) {HepPh} \\
 \end{tabular}
\vspace{-3mm}
 \caption{ {\em NDCG@k} vs.\ query time for top-$\boldsymbol{k}$
   SimRank queries on small graphs} \label{fig:exp-ndcg-small}
\vspace{-3mm}
\end{small}
\end{figure*}

 \begin{figure*}[!t]
 \begin{small}
  \centering
      \begin{tabular}{cccc}
   \multicolumn{4}{c}{\hspace{-4mm} \includegraphics[height=4.5mm]{./Figs/legend_small.eps}}
       \vspace{-1mm} \\
         \hspace{-8mm} \includegraphics[height=33mm]{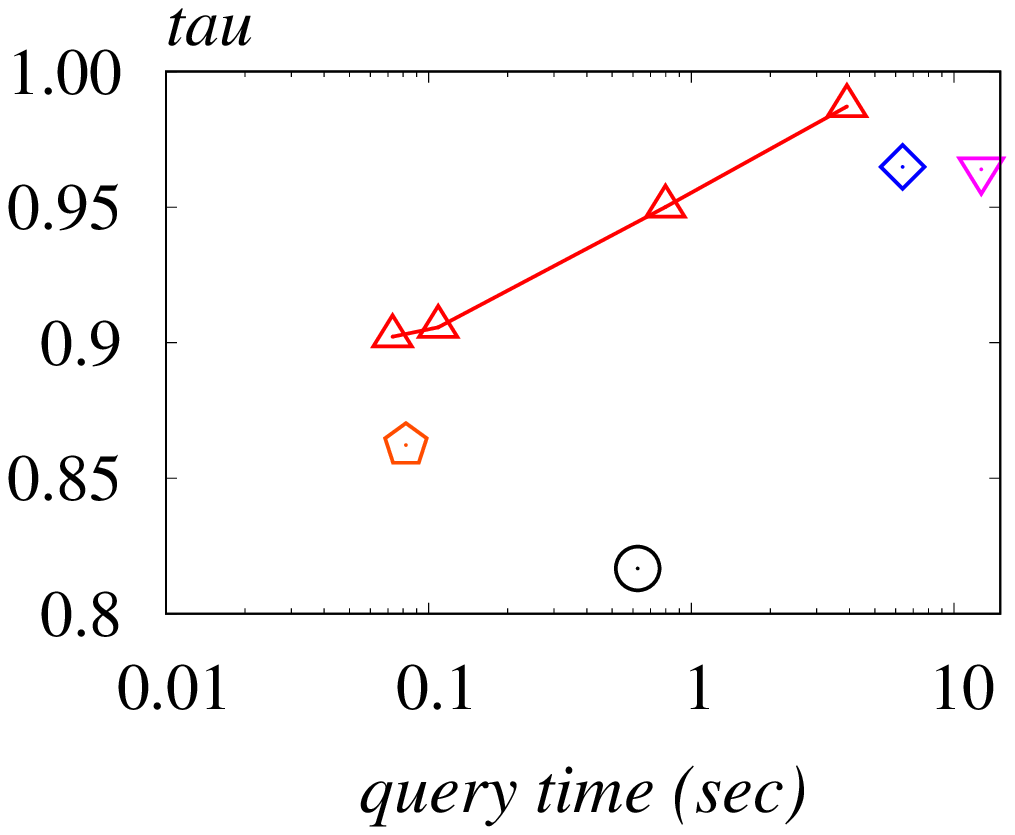}
    &
         \hspace{-8mm} \includegraphics[height=33mm]{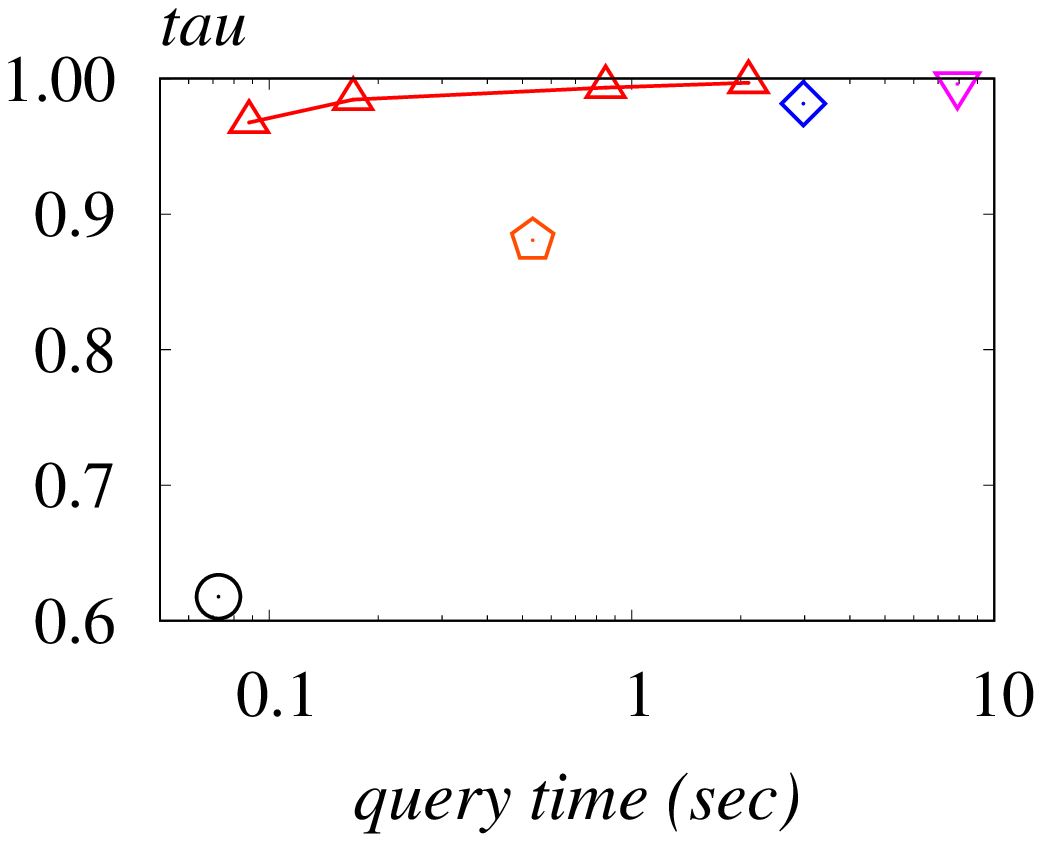}&
         \hspace{-8mm} \includegraphics[height=33mm]{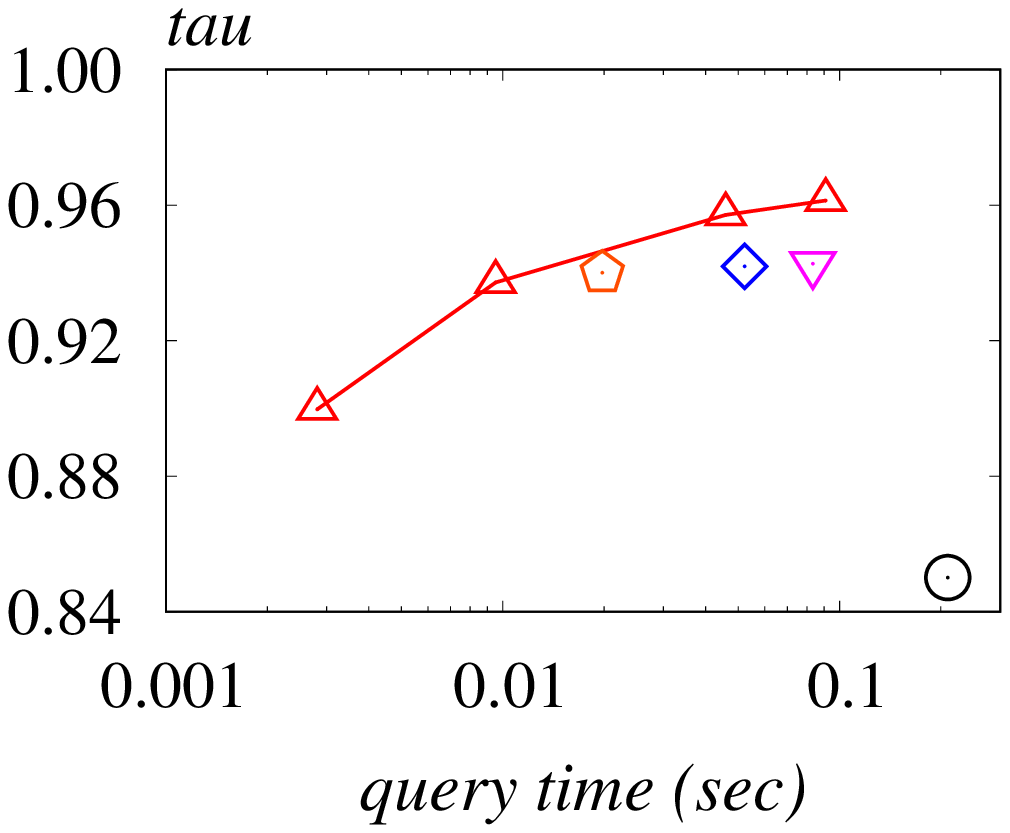} &
         \hspace{-4mm} \includegraphics[height=33mm]{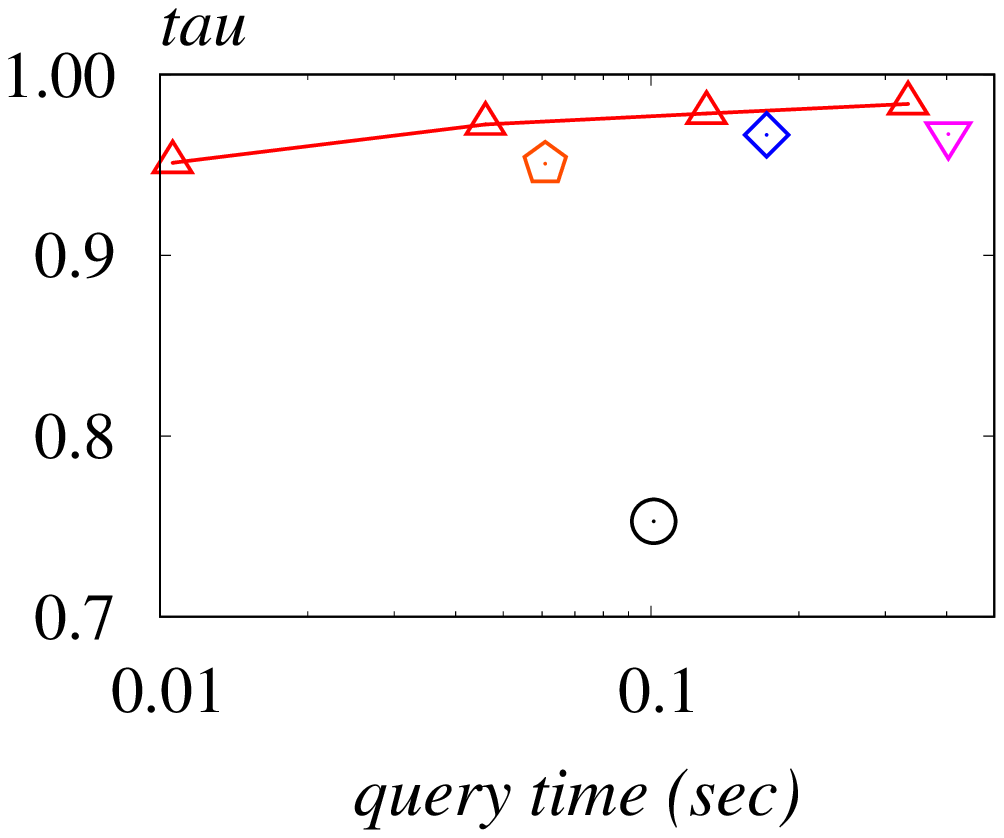}
    \vspace{-1mm} \\
        \hspace{-4mm} (a) {AS}  &
        \hspace{-8mm} (b) {Wiki-Vote} &
        \hspace{-8mm} (c) {HepTh} &
        \hspace{-8mm} (d) {HepPh} \\
  \end{tabular}
 \vspace{-3mm}
  \caption{$\tau_k$ vs. \ query time for top-$\boldsymbol{k}$ SimRank
    queries on small graphs} \label{fig:exp-tau-small}
 \vspace{-3mm}
 \end{small}
 \end{figure*}

\header
{\bf Parameters.} Following previous work
\cite{MKK14,Yu13,LVGT10,YuM15a,YuM15b}, we set the decay factor $c$ of
SimRank to 0.6.  {\em TSF} has
two internal parameters, $R_g$ and $R_q$, where $R_g$ is the number of one-way graphs stored in
the index of {\em TSF}, and $R_q$ is number of times each one-way
graph is reused in the query stage.
In accordance with the settings in~\cite{zhangexperimental} and~\cite{SLX15}, we set $R_g= 300$
and $R_q=40$. The TopSim based algorithms (i.e., {\em TopSim}, {\em Trunc-TopSim} and {\em Prio-TomSim}) has a common internal parameter $T$, which is the depth of the random walks.
{\em Trun-TopSim} has two additional parameters $h$ and $\eta$, where $1/h$
is the minimal degree threshold used to identify a high degree node
and $\eta$ is similarity threshold for trimming a random walk.
{\em Prio-TopSim} has an extra
parameter $H$, which is the number of random walks to be expanded at
each level. We set $T=3$,
$1/h = 100$, $\eta = 0.001$, and $H=100$, according to~\cite{zhangexperimental}
and~\cite{LeeLY12}. For {\em ProbeSim}, we apply all optimizations
presented in Sections \ref{sec:single-opt-prune} and
\ref{sec:single-opt-probe}. We vary the parameter $\e_a$ so that the
overall absolute error guarantee varies from $0.0125$ to $0.025$,
$0.05$, and $0.1$, so as to examine the tradeoff between the query efficiency and accuracy of {\em ProbeSim} in comparison to the other algorithms.

\header
\noindent{\bf Metrics.}
On each of the four small graphs, we use the power method
\cite{JW02} with $55$ iterations to compute the ground-truth SimRank
similarity of each node pair. This ensures that each ground-truth
value has at most $10^{-12}$ absolute error. Then, for each SimRank
similarity returned by a method, we compute its {\em absolute error (AbsError)} with
respect to the ground truth.

For each single-source SimRank query from a node $u$, we define the absolute error of the query as
$AbsError = \max_{v \in V, v \neq u} |s(u, v) - \s(u,v)|,$
which is the maximum absolute error incurred by the
method in computing the SimRank between $u$ and any other node. After that, we take the average of
the absolute error over 100 single-source SimRank queries and over 10 runs. Figure~\ref{fig:exp-max-error} shows the
average absolute error of each method as a function of its average query costs.

For Top-$k$ queries, we invoke the six algorithms to answer $100$ top-$k$ SimRank queries, with $k=50$.  We use {\em Precision@k}, the {\em Normalized Discounted Cumulative Gain (NDCG@k)}~\cite{jarvelin2000ir}, and the Kendall Tau difference $\tau_k$~\cite{nelsen2001kendall} to evaluate the accuracy of each algorithm. More precisely, given a query node $u$,
let $V_k = \{v_1, \ldots, v_k\} $ denote the top-$k$ node list
returned by the algorithm to be evaluated, and $V_k'=\{v_1', \ldots,
v_k'\}$ to be the ground truth of the top-$k$ results. {\em Precision@k}
measures  the fraction of answers that are among the ground-truth top-$k$ results, which is formally defined as
$Precision@k = {| V_k \cap V_k' | \over k}.$
{\em NDCG@k} measures the usefulness of a node based on its position in the result list, which is formally defined as
$NDCG@k = {1 \over Z_k} \sum_{i=1}^k {2^{s(u, v_i)}-1 \over \log(i+1)},$
where $Z_k = \sum_{i=1}^k {2^{s(u, v_i')-1} \over \log(i+1)}$ is the
discounted cumulative gain obtained by the ground truth of the top-$k$
results. Recall that $s(u, v_i)$ is the actual SimRank similarity between $u$ and $v_i$. Kendall Tau difference $\tau_k$ measures the accuracy of the ranking of the top-$k$ list, which
is defined as
$\tau_k = {\#(\textrm{concordant pairs})-\#(\textrm{discordant  pairs})\over k(k-1)/2}.$
Figures~\ref{fig:exp-precision-small},~\ref{fig:exp-ndcg-small}, and~\ref{fig:exp-tau-small} show the average {\em Precision@k}, {\em NDCG@k}, and $\tau_k$ of each method, respectively, as functions of its average query cost.

\header
{\bf Comparisons with TopSim based algorithms.}
Our first observation from Figure~\ref{fig:exp-max-error} is that {\em ProbSim} can achieve lower {\em AbsError} than the TopSim based algorithms, even when its query cost is much lower than those of the latter.
For example, {\em ProbeSim} yields an {\em AbsError} of 0.008 using 0.8 seconds on AS, while {\em Trun-TopSim} and {\em TopSim} achieve the same level of accuracy using 6.2 seconds and 13.5 seconds.  This is mainly due to the fact that {\em ProbeSim} is able to estimate the SimRank value up to any given precision, while the TopSim based algorithms have a level of accuracy that is equivalent to the {\em Power Method} with only $T=3$ iterations. Among the TopSim family,  the {\em AbsError} of  {\em Prio-TopSim} and {\em Trun-TopSim} are higher than that of {\em TopSim}, which concurs with the fact that the formal two algorithms use heuristics that trade accuracy for efficiency. 

For top-$k$ queries, Figures~\ref{fig:exp-precision-small} show that the query time for {\em ProbeSim} is 2 to 4 times smaller than those of the TopSim based algorithms, when providing a similar level of precision. Take the Wiki-Vote dataset for example.  {\em ProbeSim} takes less than 2 seconds to achieve a precision of $99.99\%$, while {\em TopSim} requires $8.76$ seconds.
In addition, {\em ProbeSim} achieves a precision of $99\%$ in less than 0.08 seconds, while {\em Prio-TopSim} only yields a precision of $95.5\%$ in 0.8 seconds. Meanwhile, Figure~\ref{fig:exp-ndcg-small} and~\ref{fig:exp-tau-small}  show that {\em ProbeSim} also achieves better {\em NDCG@k} and Kendall Tau difference than the TopSim based algorithms do, which suggests that the ranking of the top-$k$ results returns by {\em ProbeSim} is superior to those of the TopSim based algorithms.

\header
\noindent{\bf Comparisons with {\em TSF}.} To compare {\em ProbeSim} with {\em TSF},  we first observe from Figure~\ref{fig:exp-max-error} that the absolute error of  {\em ProbeSim} is significantly lower than that of {\em TSF}. There are two possible reasons for {\em TSF}'s relatively inferior performance. First, the number of one-way graphs $R_g$ used by {\em TSF} is only $300$, which means that the number of random walks used in the query stage is limited, leading to inaccurate SimRank estimations. In contrast, {\em ProbeSim} generates much more random walks, which enables it to achieve a much better precision. Second, as we discuss in Section~\ref{sec:compare}, {\em TSF} adopts two heuristics that make it unable to provide worst-case accuracy guarantee, which could contribute to its relatively large query error.

\begin{table*} [t]
\centering
\renewcommand{\arraystretch}{1.5}
\begin{small}
\tblcapup
\vspace{-5mm}
\caption{Space overheads and preprocessing costs comparison on large graphs.} \label{tbl:large_query}
\tblcapdown
 \begin{tabular} {|c|c|c|c|c|c||c|c|c|c|c||c|} \hline
   \multirow{2}{*}{\hspace{-2mm}{\bf Dataset}\hspace{-2mm}}  &  \multicolumn{5}{c||}{\bf
                                     Query Time (Seconds)}&
                                                            \multicolumn{5}{c||}{\bf Space Overhead (GBs)} & \multirow{2}{*}{\hspace{-1mm}{\bf Graph Size}\hspace{-1mm}} \\ \cline{2-11}
&\hspace{-1mm} {\em ProbeSim}\hspace{-1mm} & \hspace{-1mm}{\em TopSim}\hspace{-1mm} & \hspace{-1mm}{\em Trun-TopSim}\hspace{-1mm} &\hspace{-1mm} {\em Prio-TopSim}\hspace{-1mm} &\hspace{-2mm} {\em TSF} \hspace{-2mm}  &\hspace{-1mm} {\em ProbeSim} \hspace{-3mm}& \hspace{-1mm}{\em TopSim}\hspace{-1mm} &\hspace{-1mm} {\em Trun-TopSim}\hspace{-1mm} &\hspace{-1mm} {\em Prio-TopSim}\hspace{-1mm} &\hspace{-2mm} {\em TSF}\hspace{-2mm} & \\ \hline
LiveJournal & 0.4  & 554.73 & 13.31 & 1.26 & 0.52  & 0.05 & 17.2 &  10.1 & 0.09 & 10.8
                                                  & 0.88 GB \\ \hline
IT-2004& 0.006  & 4.91 & 2.34 & 0.21 & 0.93 & 0.4 & 1.7 & 0.9 & 1.0 &
                                                                      83.7
                                                                                                                                                                                                                                                                                                            & 10.9 GB\\ \hline
Twitter& 13.2 & N/A & N/A & 6220 & 175.34 & 0.6  & N/A & N/A & 1.2  &
                                                                      79.1  & 14.3 GB \\ \hline
Friendster&  3.2  & N/A
  & N/A & 58.4 &1036  &  0.9 & N/A &  N/A & 1.9 & 99.2 & 23.4 GB \\ \hline
 \end{tabular}
\end{small}
\tbldown
\vspace{-3mm}
\end{table*}

From Figures~\ref{fig:exp-precision-small}, we observe that {\em ProbeSim} dominates {\em TSF} for top-$k$ queries, as it is able to achieve higher {\em Precision@k}, {\em NDCG@k} and $\tau_k$ while incurring the same or a smaller computation overhead than {\em TSF} does. The only exception is on Wiki-Vote, where {\em ProbeSim} achieves a much higher precision ($99\%$ vs.\ $83\%$) than {\em TSF}
but incurs a slightly higher query cost (0.08 seconds vs.\ 0.06 seconds). Similar phenomenons can be observed for {\em NDCG@k} and $\tau_k$ on Wiki-Vote from Figure~\ref{fig:exp-ndcg-small}
and~\ref{fig:exp-tau-small}. This is due to the fact that the {\em AbsError} of {\em ProbeSim} is smaller than that of {\em TSF}.

An interesting observation for Wiki-Vote is that the {\em AbsError} of {\em TSF} is lower than that of {\em Prio-TopSim}, while the {\em Precision@k}, {\em NCDG@k} and $\tau_k$ of {\em Prio-TopSim} are
higher than those of {\em TSF}. One possible explanation is that Wiki-Vote is ``locally dense'' graph, in that more than $60\%$ of its nodes have zero in-degree, while the remaining ones form a dense
subgraph. Therefore, by setting  $H=100$, {\em Prio-TopSim} may omit some nodes with small SimRank values, which leads to large {\em AbsError}. However, {\em Prio-TopSim} may still examine most of the nodes with large SimRank values, thus achieving a relatively high precision.

\begin{figure*}[!t]
\begin{small}
 \centering
   \vspace{-2mm}
    \begin{tabular}{cccc}
\multicolumn{4}{c}{\hspace{-4mm} \includegraphics[height=5mm]{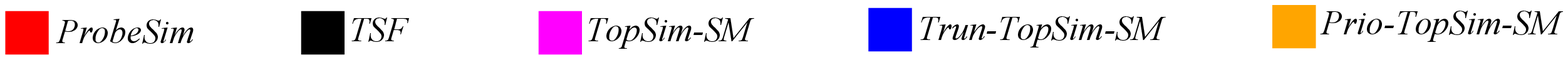}} \vspace{-1mm} \\
        \hspace{-4mm} \includegraphics[height=30mm]{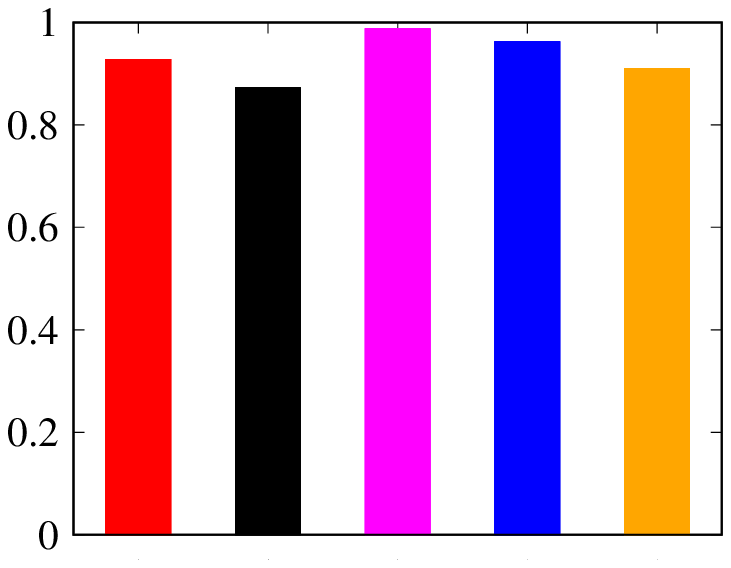} &
        \hspace{-4mm} \includegraphics[height=30mm]{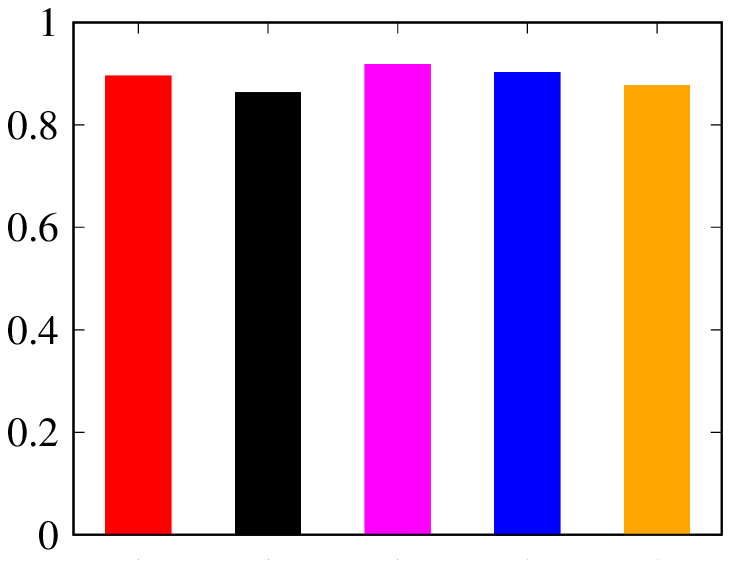} &
        \hspace{-4mm} \includegraphics[height=30mm]{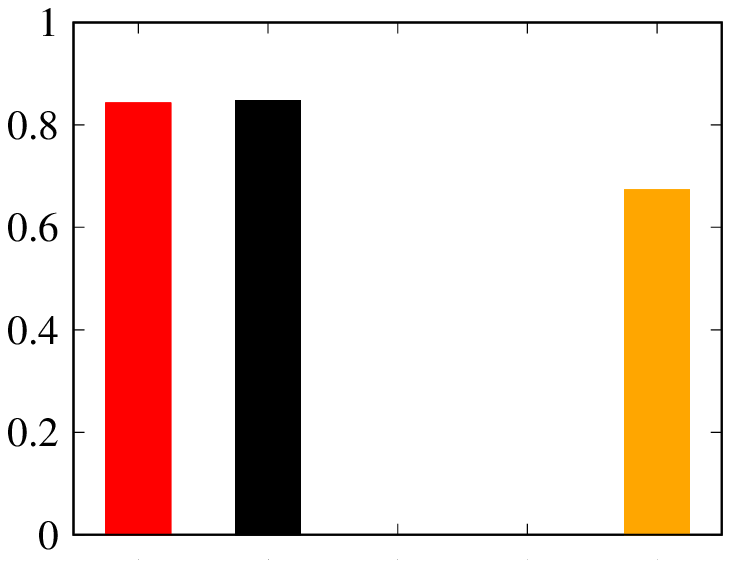} &
        \hspace{-4mm} \includegraphics[height=30mm]{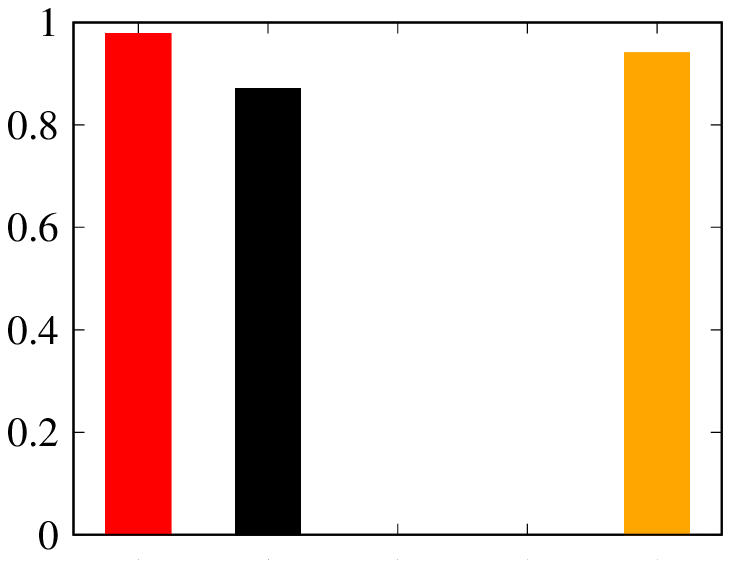}
   \vspace{-1mm} \\
       \hspace{-4mm} (a) {LiveJournal}  &
       \hspace{-8mm} (b) {IT-2004} &
       \hspace{-8mm} (c) {Twitter} &
       \hspace{-8mm} (d) {Friendster} \\
 \end{tabular}
\vspace{-3mm}
 \caption{ {\em Precision@k} for top-$\boldsymbol{k}$ SimRank queries on large graphs} \label{fig:exp-precision-large}
\vspace{-0mm}
\end{small}
\end{figure*}

\begin{figure*}[!t]
\begin{small}
 \centering
   \vspace{-2mm}
    \begin{tabular}{cccc}
\multicolumn{4}{c}{\hspace{-4mm} \includegraphics[height=5mm]{./Figs/legend_large.eps}} \vspace{-1mm} \\
        \hspace{-4mm} \includegraphics[height=30mm]{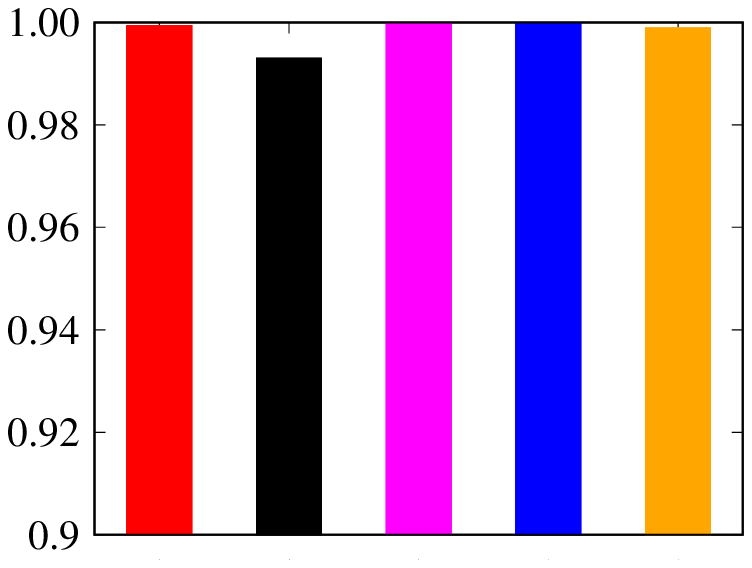} &
        \hspace{-4mm} \includegraphics[height=30mm]{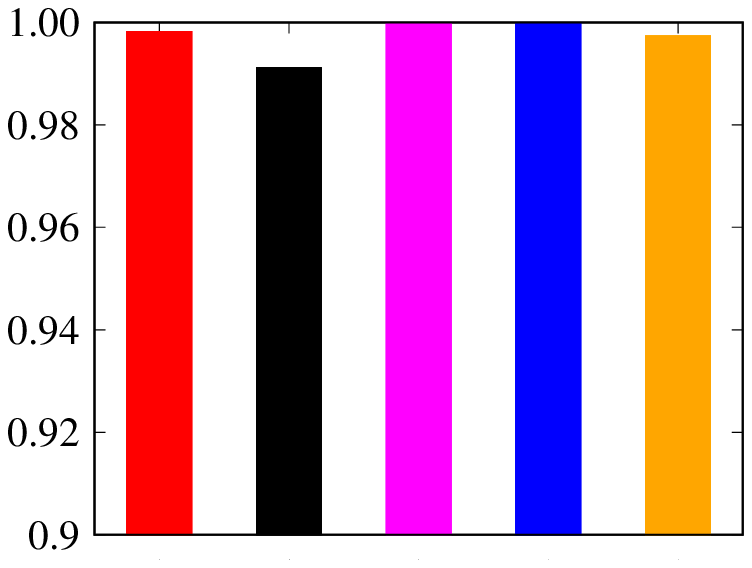} &
        \hspace{-4mm} \includegraphics[height=30mm]{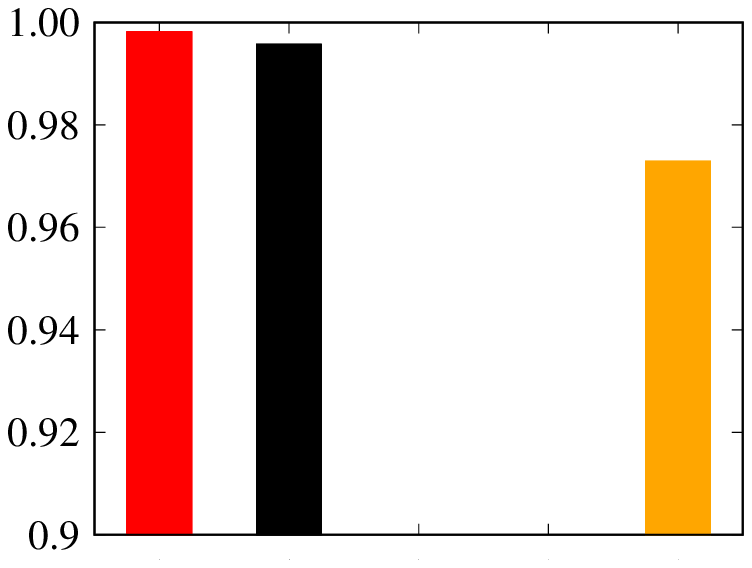} &
        \hspace{-4mm} \includegraphics[height=30mm]{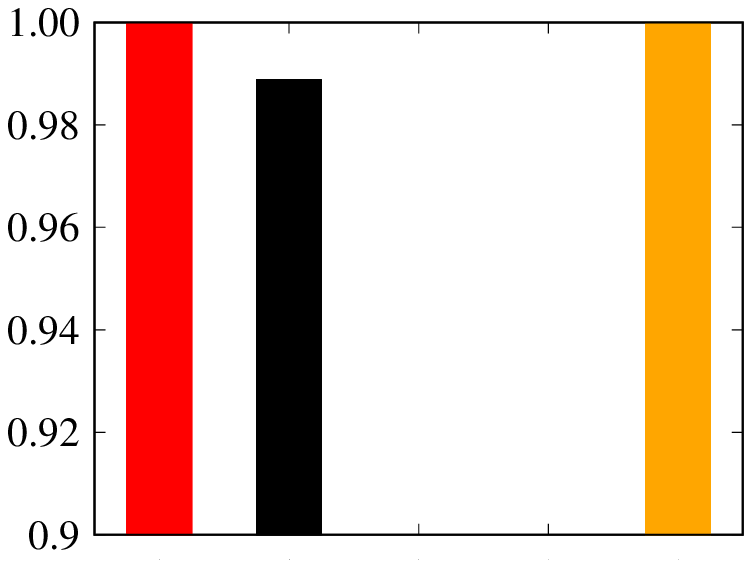}
   \vspace{-1mm} \\
       \hspace{-4mm} (a) {LiveJournal}  &
       \hspace{-8mm} (b) {IT-2004} &
       \hspace{-8mm} (c) {Twitter} &
       \hspace{-8mm} (d) {Friendster} \\
 \end{tabular}
\vspace{-3mm}
 \caption{ {\em NDCG@k} for top-$\boldsymbol{k}$ SimRank queries on
   large graphs} \label{fig:exp-ndcg-large}
\vspace{-3mm}
\end{small}
\end{figure*}

 \begin{figure*}[!t]
 \begin{small}
  \centering
     \begin{tabular}{cccc}
 \multicolumn{4}{c}{\hspace{-4mm} \includegraphics[height=5mm]{./Figs/legend_large.eps}} \vspace{-1mm} \\
         \hspace{-4mm} \includegraphics[height=30mm]{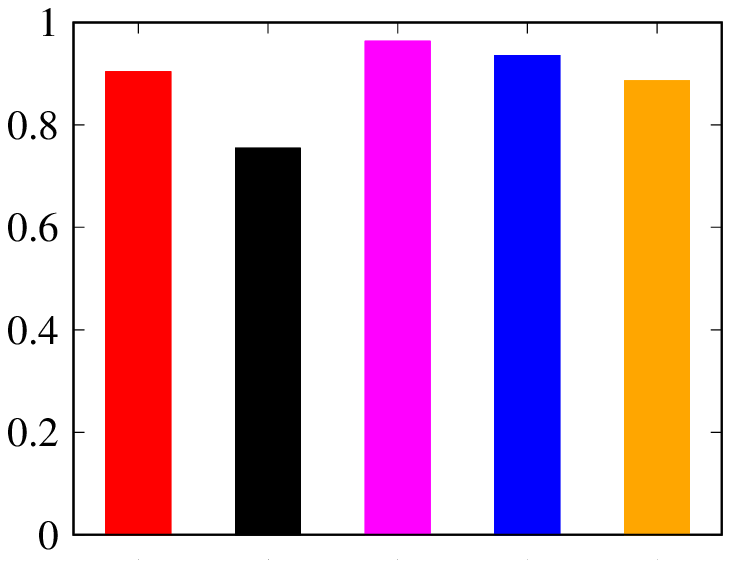} &
         \hspace{-4mm} \includegraphics[height=30mm]{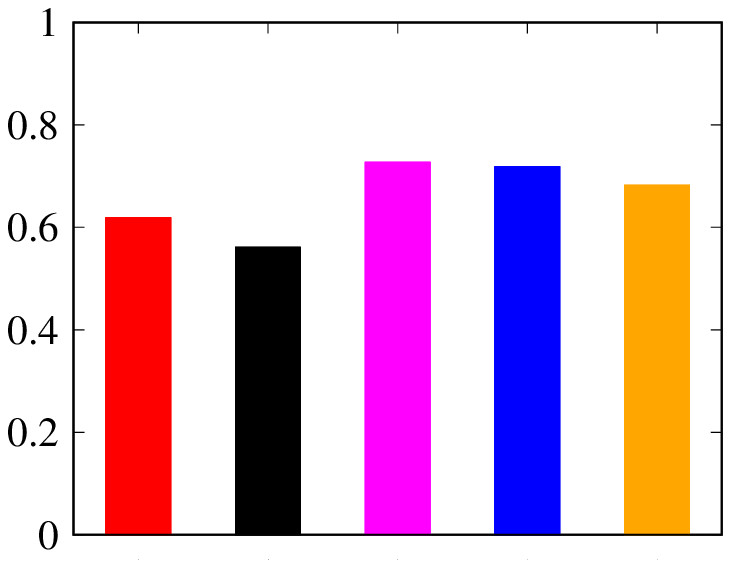} &
         \hspace{-4mm} \includegraphics[height=30mm]{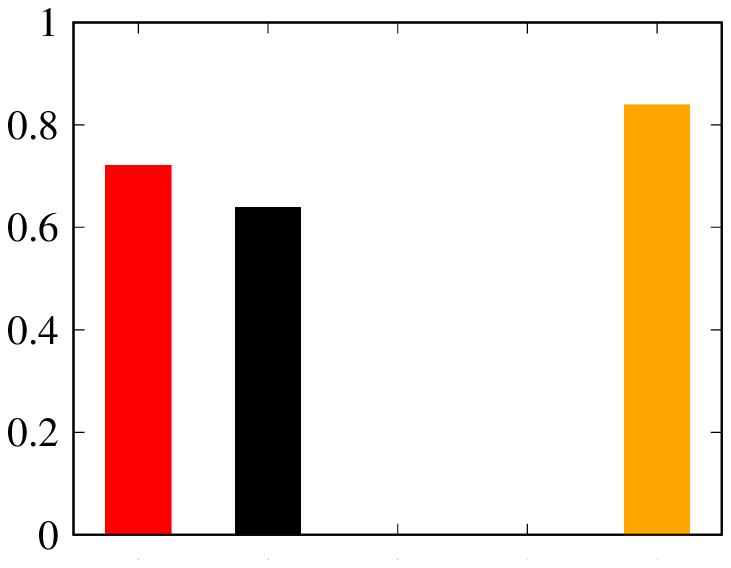} &
         \hspace{-4mm} \includegraphics[height=30mm]{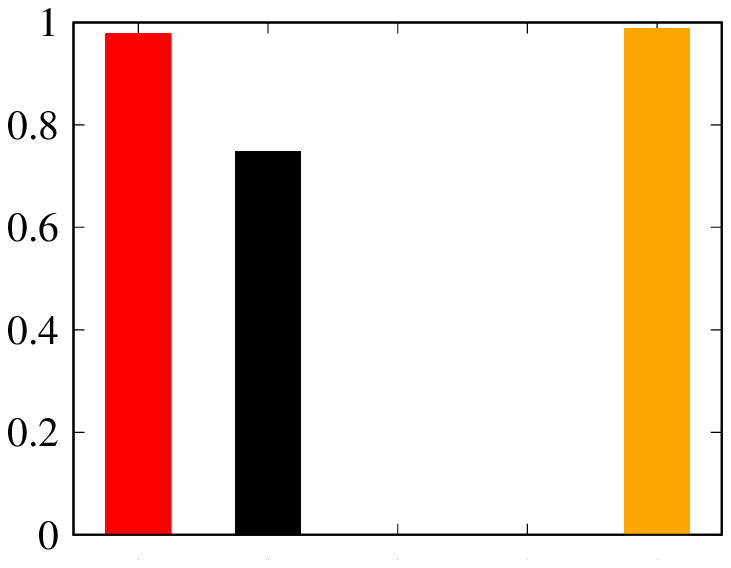}
    \vspace{-1mm} \\
        \hspace{-4mm} (a) {LiveJournal}  &
        \hspace{-8mm} (b) {IT-2004} &
        \hspace{-8mm} (c) {Twitter} &
        \hspace{-8mm} (d) {Friendster} \\
  \end{tabular}
 \vspace{-3mm}
  \caption{ $\tau_k$ for top-$\boldsymbol{k}$ SimRank queries on large graphs} \label{fig:exp-tau-large}
 \vspace{-3mm}
 \end{small}
 \end{figure*}

\vspace{-1mm}
\subsection{Experiments on Large Graphs}

Next, we evaluate the algorithms on the four large graphs with up to
$2.59$ billions of edges. Previous work ignores the accuracy
comparisons on such graphs, as the ground truth of top-$k$ results on these graphs are unavailable due to the high computational cost of the {\em Power Method}. To the best of our knowledge, we are the first to empirically evaluate both accuracy and efficiency
of SimRank algorithms on billion-edge graphs.


\vspace{-1mm}
\header
\noindent{\bf Pooling.} As {\em Power Method} only works for small graphs, we need an alternative approach to evaluate the accuracy of SimRank algorithms on large graphs.
Towards this end, we use {\em pooling}~\cite{manning2008introduction}, which is a standard approach for evaluating top-$k$ documents ranking quality in Information Retrieval (IR) systems when the ground-truth ranking scores of all documents are difficult to obtain.  The basic idea of pooling is as follows. Suppose that we are to evaluate $\ell$ IR systems, $A_1, \ldots, A_\ell$, each of which aims to return the top-$k$ documents that are most relevant to a certain query. We first take the top-$k$ documents returned by each system, and we merge them into a pool, with duplicates removed. 
Then, we present the results in the pool to experts for evaluation. Based on the relevance scores provided by the experts, we pick the best $k$ documents from the pool, and use them as the ground truth
for evaluating the top-$k$ results returned by $A_1, \ldots, A_\ell$.

In the scenario of evaluating SimRank algorithms, we use single-pair Monte Carlo algorithm as the ``expert'' for gauging the results in the pool. More precisely, for each query node $u$, we retrieve the top-$k$ nodes returned by each algorithm, remove the duplicates, and merge them into a pool. For each node $v$ in the pool, we estimate $s(u, v)$ using the Monte Carlo algorithm. We set the parameters of the Monte Carlo algorithms such that it incurs an error less than $0.0001$ with a confidence over $99.999\%$. Then we take the $k$ nodes with the highest estimated SimRank scores from the pool as the ground truth. Essentially, these $k$ results are the best possible $k$ nodes that can be obtained by any of the algorithms considered. 

\vspace{-1mm}
\header
\noindent{\bf Parameters and setups. }
We compare {\em ProbeSim}, {\em TopSim}, {\em Trun-TopSim}, {\em Prio-TopSim} and {\em TSF} using the pooling approach.  On each dataset, we select $20$ nodes
uniformly at random from the nodes with nonzero in-degrees, and we generate top-$k$ SimRank queries from each node.  We generate top-$k$ SimRank queries from each node to evaluate the algorithms.
For {\em TSF} and the TopSim based algorithms, we use the same parameters as in the experiments on small graphs. For {\em ProbeSim} , however, we can no longer vary the error parameter
$\e_a$, since changing its parameters may result in a different ``ground truth'' top-$k$ nodes in the pool, rendering it difficult to compare different algorithms. 
Therefore, we fix $\e_a=0.1$ for {\em ProbeSim} in this set of experiments.

Table~\ref{tbl:large_query} shows the average query time of each algorithm, while Figure~\ref{fig:exp-precision-large}, \ref{fig:exp-ndcg-large}, and~\ref{fig:exp-tau-large} show the {\em Precision@k}, {\em NDCG@k}, and $\tau_k$, respectively, of each algorithm. We exclude {\em TopSim} and {\em Trun-TopSim} from the experiments on Twitter and Friendster, because for some queries they either run out of memory or require more than 24 hours. In addition, on Friendster, the index size of {\em TSF} exceeds the size of the main memory (i.e., 96GB), due to which we move $100$
one-way graphs in {\em TSF} to the disk, as suggested in~\cite{SLX15}. 

\vspace{-1mm}
\header
\noindent{\bf Comparisons with TopSim based algorithms.}
Our first observation from Table~\ref{tbl:large_query} is that
 the query cost of {\em ProbeSim}
is significantly lower than those of the TopSim based
algorithms on all four graphs. In particular, on the Twitter dataset, {\em TopSim} and {\em
  Trun-TopSim} could not finish query processing in 24 hours, while {\em Prio-TopSim} takes 2 hours on average to answer a query. On the other hand, {\em ProbeSim} is able to answer a query within 13 seconds on average.

It has been observed in~\cite{zhangexperimental} that the running time of TopSim based algorithms is sensitive to the subgraph structure and density around the query node. For example, the running time of {\em Prio-TopSim} on ``locally dense'' graphs, such as
Twitter and Friendster, are higher than those on ``locally sparse''
graphs, such as IT-2004. Table~\ref{tbl:large_query} suggests that {\em
  ProbeSim} shares the same property, as its query cost on Twitter is significantly higher than those on other graphs. In contrast, the query cost of ProbeSim is less sensitive to the local subgraph structure and density. 

Figure~\ref{fig:exp-precision-large}, \ref{fig:exp-ndcg-large}, and
~\ref{fig:exp-tau-large} show  that in terms of accuracy, {\em
  ProbeSim} outperforms the TopSim based algorithms on Twitter and
Friendster. For example, {\em
  ProbeSim} achieves an average precision of $84\%$ on Twitter, while
{\em Prio-TopSim} only achieves an average precision of $67\%$.  On
LiveJournal and IT-2004,  {\em TopSim} and {\em Trun-TopSim} offer
a slightly better accuracy than {\em ProbeSim} does, at the cost of significantly higher cost. For example, {\em TopSim} yields a precision of $91.7\%$ on IT-2004, while {\em ProbeSim}
achieves a precision of $89.5\%$; however, the running times of {\em TopSim} is 600 times higher than that of {\em ProbeSim}
on IT-2004. Figures~\ref{fig:exp-ndcg-large} shows that {\em NDCG@k} of {\em TopSim} and {\em ProbeSim} are essentially the same on IT-2004, which indicates that the accuracy of the two methods are highly comparable.



\vspace{-1mm}
\header
\noindent{\bf Comparisons with {\em TSF}.}
 From Table~\ref{tbl:large_query}, we first observe that {\em ProbeSim} achieves better query time comparing to {\em TSF}. For example, it takes 180 seconds for {\em TSF} to answer a query on Twitter, while {\em ProbeSim} only requires $12$ seconds. Furthermore, {\em TSF} runs out of 96 GB memory on Friendster, even though the dataset itself is only 23 GB in size. Consequently, {\em TSF} has to store some of the one-way graphs on the disk, which leads to severe degradation of query efficiency. 
 On the other hand, {\em ProbeSim} is able to handle a query on Friendster within 3.2 seconds on average.

LiveJournal is the only dataset on which  {\em ProbeSim} and {\em TSF} use the same amount of time to answer a query. The main reason is that (i) LiveJournal (with only 4 million nodes) is a relatively small graph comparing to the other three large graphs, and (ii) the query cost of {\em TSF} tends to decrease when the graph size reduces. In contrast, {\em ProbeSim} is less sensitive to
the graph size.

In terms of query accuracy, Figure~\ref{fig:exp-precision-large}, ~\ref{fig:exp-ndcg-large}, and~\ref{fig:exp-tau-large} show that {\em ProbeSim} is able to
provide more accurate results than {\em TSF} does on LiveJournal,IT-2004, and Friendster. For example, on Friendster, {\em ProbeSim} achieves a {\em Precision@k} of $98\%$ and an {\em NDCG@k} of $0.9998$, while {\em TSF} yields a {\em Precision@k} of $87\%$ and an {\em NDCG@k} of $0.9914$. On Twitter, {\em ProbeSim} and {\em TSF} offer almost the same {\em Precision@k} and {\em NDCG@k}, but {\em ProbeSim} outperforms {\em TSF} on the Kendall Tau difference. This suggests
that the ranking accuracy of {\em ProbeSim} is better than that of {\em TSF}. We also observe that the precision of all three algorithms are relatively low on Twitter, which implies that SimRank on Twitter is still a difficult problem to solve.


An interesting observation is that {\em TSF} outperforms {\em Prio-TopSim} in terms of accuracy on Twitter, which contrasts the case on small graphs. As mentioned in~\cite{SLX15}, {\em Prio-TopSim} only expands $H$ random walks at each level, and hence, its performance heavily rely on the random walks chosen. Because Twitter has very denser local structures, $H=100$ may not be not sufficient for {\em Prio-TopSim} to explore all possible candidates. In contrast, the random sample framework of {\em TSF} treats each node equally, and thus, it
gives relatively stable performance across social graphs and web graphs.

Finally, Table~\ref{tbl:large_query} shows that {\em TSF} incurs significant overheads in terms of space. In particular, the index size of {\em TSF} is one to two orders of magnitude larger than the size of the input graph $G$. For example, the index of {\em TSF} for IT-2004 is 85 GB, while the graph size is only 8 GB. Furthermore, {\em TSF} runs out of memory when preprocessing Friendster, which leads to performance degradation. In contrast, {\em ProbeSim} is able to efficiently handle queries on all large graphs without any preprocessing.

\vspace{-2mm}
\section{Conclusions} \label{sec:conclusions}
\balance
This paper presents {\em ProbeSim}, an algorithm for single-source and top-$k$ SimRank computation without preprocessing. {\em ProbeSim} answers any single-source SimRank query in $O(\frac{n}{{\e_a}^2} \log \frac{n}{\delta})$ expected time, and it ensures that, with $1-\delta$ probability, all SimRank similarities returned have at most $\e_a$ absolute error. Our experiments show that the algorithm significantly outperforms the existing approaches in terms of query efficiency, and they are more scalable than the existing index-based methods, as they are able to handle graphs that are too large for the latter to preprocess. For future work, we plan to study lightweight indexing approaches for SimRank that provide higher effectiveness than our current algorithms on large graphs (such as Twitter) without incurring significant space and time in computation.



\vspace{-2mm}
\section{Acknowledgments}
This work was partly supported by the National Natural Science Foundation of China
      (No. 61502503, No. 61532018, No. 61502324 and No. 61472427), by the 
      National Key Basic Research Program (973 Program) of China (No.
      2014CB340403, No. 2012CB316205), by Academy of Finland (310321),
      by the DSAIR center at the Nanyang Technological University, by
      a gift grant from Microsoft Research Asia, and by Grant
      MOE2015-T2-2-069 from MOE, Singapore. 

%


\allowdisplaybreaks

\bibliographystyle{abbrv}
\bibliography{ref}

\appendix
\section{Chernoff Bound} \label{sec:chernoff}
\begin{lemma}[Chernoff Bound \cite{ChungL06}] \label{lmm:chernoff}
For any set $\{x_i\}$ ($i \in [1, n_x]$) of i.i.d.\ random variables with mean $\mu$ and $x_i \in [0, 1]$,
$$\Pr\left\{\left|\sum_{i=1}^{n_x} x_i - n_x \mu\right| \geq n_x \e\right\} \leq \exp\left(-\dfrac{n_x \cdot \e^2}{\frac{2}{3}\e + 2\mu}\right).$$
\end{lemma}

\section{Proofs} \label{sec:proofs}

\subsection{Proof of Theorem~\ref{thm:pruning}}
To prove Theorem~\ref{thm:pruning}, we independently bound the error
introduced by pruning rule 1 and pruning rule 2.
\begin{lemma}
Let $\s_k(u,v, \varepsilon_t)$ and $\s_k(u,v)$ denote
estimator with and without applying pruning rule 1, respectively. We have
$$0 \le \s_k(u,v) - \s_k(u,v, \varepsilon_t) \le \varepsilon_t.$$
\end{lemma}
\begin{proof}
Let $W(u) = (u_1, \ldots, u_\ell)$ denote the
original $\sqrt{c}$-walk, and  $W(u, \ell_t) = (u_1, \ldots,
u_{\ell_t})$ denote the $\sqrt{c}$-walk after truncation.
For each vertex $x \in V$,  $ i = 2, \ldots, \ell$ and $j = 0,
\ldots,i-1$, let
$Score_i(x, j, \varepsilon_t)$ and $Score_i(x, j)$ denote the score computed
by \textsf{PROBE}($W(u,i)$) after the $j$-th iteration, with and
without truncation, respectively.
 We inductively prove that
1) For each $d = 0, \ldots, \ell-\ell_t$, and for any $x\in V$
$$0 \le \sum_{j=0}^d Score_{\ell-d+j}(x, j) - \sum_{j=0}^d
Score_{\ell-d+j}(x, j, \varepsilon_t) \le 1.$$
2) For each $d = \ell-\ell_t, \ldots, \ell$, and for any $x\in V$
$$0 \le \sum_{j=0}^d Score_{\ell-d+j}(x, j) - \sum_{j=0}^d
Score_{\ell-d+j}(x, j, \varepsilon_t) \le
(\sqrt{c})^{d-\ell+\ell_t}.$$

For the base case, notice that there is only one vertex $v_\ell$, and
$$0 \le Score_{\ell}(v_\ell, 0) -
Score_{\ell}(v_\ell, 0, \varepsilon_t) \le 1.$$
Suppose 1) is true for $d$. For $d+1 \le \ell-\ell_t$, each $x \in V$,
$x \neq u_{\ell-d-1}$ satisfies that
$\sum_{j=0}^d Score_{\ell-d+j}(x, j, \varepsilon_t)=0$ and each
$$\sum_{j=0}^{d+1} Score_{\ell-d+1+j}(x, j) = \sum_{y \in I(x)}
\sum_{j}^d Score_{\ell-d+j}(y, j) \cdot {\sqrt{c} \over |I(x)|} .$$
By the induction hypothesis, we have $\sum_{j}^d Score_{\ell-d+j}(y,
j) \le 1$, and thus
$$\sum_{j=0}^{d+1} Score_{\ell-d+1+j}(x, j) \le \sqrt{c}|I(x)| /
|I(x)| = \sqrt{c} \le 1.$$
We also note $\sum_{j=0}^d Score_{\ell-d+j}(u_{\ell-d-1}, j,
\varepsilon_t)=0$ since $u_{\ell-d-1}$ is truncated,  and $\sum_{j=0}^d Score_{\ell-d+j}(u_{\ell-d-1},
j)=1$ since we do not add scores to $u_{\ell-d-1}$ at this step. Thus
claim 1 follows.

Similarly, for claim 2, we use induction proof.  The base case $d =
\ell-\ell_t$ follows from claim 1.  Assume the claim holds for $d$,
and consider $d+1$. For each $x \in V$,
$x \neq u_{\ell-d-1}$ satisfies that
$\sum_{j=0}^d Score_{\ell-d+j}(x, j, \varepsilon_t)=0$ and each
$$\sum_{j=0}^{d+1} Score_{\ell-d+1+j}(x, j) = \sum_{y \in I(x)}
\sum_{j}^d Score_{\ell-d+j}(y, j) \cdot {\sqrt{c} \over |I(x)|} .$$
By the induction hypothesis, we have $\sum_{j}^d Score_{\ell-d+j}(y,
j) \le (\sqrt{c})^{d-\ell+\ell_t}$, and thus
$$\sum_{j=0}^{d+1} Score_{\ell-d+1+j}(x, j) \le \sqrt{c} \cdot
(\sqrt{c})^{d-\ell+\ell_t}{ |I(x)|  \over
|I(x)| }= (\sqrt{c})^{d+1-\ell+\ell_t} .$$

By  setting $d=\ell_t$ in claim 2, we have
$$0 \le \sum_{j=1}^\ell Score_{j}(x, j) - \sum_{j=1}^\ell
Score_{j}(x, j, \varepsilon_t)  \le (\sqrt{c})^{\ell_t} \le
\varepsilon_t,$$
and the Lemma follows.
\end{proof}

\subsection{Proof of Theorem~\ref{thm:randomized_error}}
\begin{lemma}
\label{lem:unbiasness_random}
Fix a reverse path $W(u, i) = ( u_1, \ldots, u_i)$ and an
node $v \in V, v \neq u$. The randomized \textsf{PROBE} algorithm
on $(u_1, \ldots, u_i)$  set $Score(v) =1$ with probability $P(v,
W(u, i))$, the first-meeting probability of $v$ with
respect to reverse path $W(u, i)$.
\end{lemma}
\begin{proof}
Recall that $P(v, W(u, i))$ is the probability that a random $\sqrt{c}$-walk $W(v) = (v_1, \ldots, v_i,
  \ldots)$  and the reverse path $W(u, i) = (u_1, \ldots, u_i)$ first
meet at $u_i = v_i$.

Similar to the proof of Lemma~\ref{lem:Score}
We prove the following claims: for $j = 0 ,\ldots, i-1$,  after the
$j-1$-th iteration in Algorithm~\ref{alg:RProbe}, a vertex
$v \in V$  is selected into $\mathcal{H}_{j}$ with probability $P(v, (u_{i-j},
\ldots, u_i))$. Note the this claim implies that after $i$ iterations, a vertex
$v \in V$  is selected into $\mathcal{H}_{i-1}$ with probability $P(v, (u_{1},
\ldots, u_i)) = P(v, W(u, i))$, and the Lemma will follow.

The proof is done by induction. Level $0$ contains a single vertex
$u_i$, and we have $u_i \in \mathcal{H}_0$ with
probability $1$. Since $P(v,(u_i)) = \Pr[v_i= u_i]=1$, the claim holds. Assume that
the claim holds for all vertices after the $(j-1)$-th iteration. After
the $j$-th iteration, consider a vertex
$v \in V$. By Algorithm~\ref{alg:RProbe}, $v$ is selected into $\mathcal{H}_{j+1}$ if
and only if there
exists an $x \in V$, such that 1) $x \in \mathcal{H}_j$; 2) $(x, v) \in I(v)$ is
selected in $I(v)$; 3) $v$ is chosen with probability
$\sqrt{c}$.
By the induction hypothesis, we
know each $x \in I(v)$ is selected into $\mathcal{H}_j$
with probability $P(x, (u_{i-j}, \ldots, u_i))$. It follows that $v$ is selected
into $\mathcal{H}_{j+1}$ with probability
\begin{align}
\sum_{x \in I(v)}{\sqrt{c} \over |I|} \cdot  P(x, (u_{i-j}, \ldots,
u_i)) = P(x, (u_{i-j-1}, \ldots, u_i)). \label{eqn:score_random}
\end{align}
where equation~\eqref{eqn:score_random} follows from the proof of
Lemma~\ref{lem:Score}. Therefore, the claim is true,
and  after $i$ iterations, a vertex
$v \in V$  is selected into $\mathcal{H}_{i-1}$ with probability $P(v, (u_{1},
\ldots, u_i)) = P(v, W(u, i))$.
\end{proof}

By Lemma~\ref{lem:unbiasness_random}, we can use the randomized \textsf{PROBE} algorithm in
Algorithm~\ref{alg:single_source}, which would give us an algorithm
that runs in $O(\frac{n}{{\varepsilon_a}^2} \log
\frac{n}{\delta})$ time while retaining the error guarantee in
Theorem~\ref{thm:EScoreAbsError}. However, as we shall see, we can
combine the randomized and deterministic \textsf{PROBE} algorithms to
achieve both worst-case and real-world efficiency.
Next we show that the error introduced by pruning rule 2 is bounded
by $\varepsilon_p \ell_t$.

\begin{lemma}
For each vertex $x \in V$,  $ i = 2, \ldots, \ell$, let
$Score_i(x, \varepsilon_p)$ and $Score_i(x)$ denote the score computed
by \textsf{PROBE}($W(u,i)$), with and
without pruning, respectively. We have
$$0 \le Score_i(x)  - Score_i(x, \varepsilon_p) \le
\varepsilon_p .$$
\end{lemma}
\begin{proof}
For each vertex $x \in V$,  $ i = 2, \ldots, \ell$ and $j = 0,
\ldots,i-1$, let
$Score_i(x, j, \varepsilon_p)$ and $Score_i(x, j)$ denote the score computed
by \textsf{PROBE}($W(u,i)$) after the $j$-th iteration, with and
without pruning, respectively. We inductively prove that for each $j = 0,
\ldots,i-1$,
$$0 \le Score_i(x, j)  - Score_i(x, j, \varepsilon_p) \le
\varepsilon_p / (\sqrt{c})^{i-j-1}.$$
More precisely, for the base case $j=0$, $\mathcal{H}_0$ contains a
single vertex $u_i$, and $Score_i(u_i, 0) =  Score_i(u_i, 0,
\varepsilon_p) =1$. Assuming the claim holds for $j$. For $j+1$,
notice that for each vertex $x \in V$,
$$Score_i(x, j+1)  = \sum_{y \in
  I(x)}Score_i(y, j) \cdot {\sqrt{c} \over |I|},$$
and by the induction hypothesis, we have
$$0 \le Score_i(y, j) - Score_i(y,
j, \varepsilon_p) \le \varepsilon_p /(\sqrt{c})^{i-j-1}.$$
It follows that
\begin{align*}
Score_i(x, j+1) - Score_i(x, j+1, \varepsilon_p) &\le {
\varepsilon_p \over (\sqrt{c})^{i-j-1}} \cdot \sqrt{c} \\
&= { \varepsilon_p \over (\sqrt{c})^{i-j-2}}.
\end{align*}
By setting $j=i-1$, we have
$$0 \le Score_i(y,i-1) - Score_i(y,
i-1, \varepsilon_p) \le \varepsilon_p.$$
Since $Score_i(y) =Score_i(y,i-1)$ and  $Score_i(y) = Score_i(y,
i-1, \varepsilon_p)$, the Lemma follows.
\end{proof}

\begin{proof}[of Theorem\ref{thm:pruning}]
We only need to show that with probability $1-\delta/n$, the error
contributed by pruning rule 2 is at most ${1+\varepsilon \over 1-\sqrt{c}}
  \cdot \varepsilon_p$.
Consider the $k$-th $\sqrt{c}$-walk $W_k(u) = (u_1, \ldots, u_{\ell_k})$ of length
$\ell_k-1$, summing up all $i=2, \ldots,\ell$ follows that
$$0 \le \sum_{i=1}^{\ell_k}Score_i(y,i-1) - \sum_{i=1}^{\ell_k} Score_i(y,
i-1, \varepsilon_p) \le \varepsilon_p \ell,$$
which indicates that $0 \le \s_k(u,v) -\s_k(u,v, \varepsilon_p) \le
\varepsilon_p \ell_k$. Recall that let $\s_k(u,v,
\varepsilon_p) $ and $\s(u, v)$ denote the final estimators with and without pruning,
respectively. We have $\s(u,v,
\varepsilon_p)  = {1\over n_r} \sum_{k=1}^{n_r} \s_k(u,v,
\varepsilon_p) $, and $\s(u,v)  = {1\over n_r} \sum_{k=1}^{n_r} \s_k(u,v) $.
It follows that
\begin{align*}
\s(u,v) - \s(u,v, \varepsilon_p)  &= {1\over n_r} \sum_{k=1}^{n_r} \s_k(u,v) - {1\over n_r} \sum_{k=1}^{n_r} \s_k(u,v,
\varepsilon_p) \\
&= {\varepsilon_p\over n_r} \sum_{k=1}^{n_r} \ell_k.
\end{align*}
We note that the $\ell_k$'s are i.i.d. random variable with
expectation at most $\mu = {1\over 1-\sqrt{c}}$. By Chernoff bound, we
have
\begin{equation*}
\begin{aligned}
\Pr[|\sum_{k=1}^{n_r} \ell_k -\mu n_r | \ge \varepsilon_s \mu  n_r]
&\le \exp\left(-  {\varepsilon_s^2 \mu n_r   \over 3} \right)\\
& \le
\exp\left(-  c \mu \log {n \over \delta}  \right).
\end{aligned}
\end{equation*}
Note that $\mu = {1\over 1-\sqrt{c}} \ge { 1 \over c}$ for $c \ge 0.5$, we have
$$\Pr\left [{1\over n_r }\sum_{k1}^{n_r} \ell_k \le \mu  +
  \varepsilon_s \mu = {1+\varepsilon \over 1-\sqrt{c}} \right]  \ge 1- {\delta \over n},$$
and thus
$$\Pr\left [ \forall v \in V, \s(u,v) - \s(u,v, \varepsilon_p) \le  {1+\varepsilon \over 1-\sqrt{c}}
  \cdot \varepsilon_p\right]  \ge 1- {\delta \over n},$$
and the Theorem follows.
\end{proof}

\end{sloppy}
\end{document}